\documentclass[journal]{IEEEtran}
\IEEEoverridecommandlockouts

\usepackage{amsmath,amssymb,amsfonts,amsthm}
\usepackage{inputenc}
\usepackage{enumerate}
\usepackage{graphicx}
\usepackage{multicol}
\usepackage{multirow}
\newtheorem{thm}{Theorem}
\newtheorem{prop}{Proposition}
\newtheorem{lem}{Lemma}
\newtheorem{cor}{Corollary}

\newtheorem{exam}{Example}

\newtheorem{rem}{Remark}

\def\0{{\mathbf 0}}

\newcommand{\F}{\mathbb{F}}

\begin{document}

\sloppy

\title{Hulls of special typed linear codes and constructions of new EAQECCs}
\author
{ 
Lin Sok, email: soklin\_heng@yahoo.com
}

\date{}



\maketitle


\begin{abstract}
In this paper, we study Euclidean and Hermitian hulls of generalized Reed-Solomon codes and twisted generalized Reed-Solomon codes, as well as the Hermitian hulls of Roth-Lempel typed codes. 
We present explicit constructions of MDS and AMDS linear codes for which their hull dimensions are well determined. As an application, we provide several classes of entanglement-assisted quantum error correcting codes with new parameters.

\end{abstract}

\section{Introduction}

MDS codes are an interesting class of codes due to their capability of the largest error correction. The well-known class of MDS codes is that of Reed-Solomon (RS) codes and generalized Reed-Solomon (GRS) codes.
Twisted Reed-Solomon (TRS) codes and twisted generalized Reed-Solomon (TGRS) codes, which have been proven to be not equivalent to RS codes or GRS codes, have recently been introduced by Beelen {\em et al.} \cite{BeePuRos17} and have an application in code based crypto-systems, for instance, the well-known McEliece's crypto-systems, due to their decoding efficiency. Beelen {\em et al.} \cite{BeePuRos18} proposed a decoding algorithm for such codes.
The authors \cite{HuangYueNiuLi} studied the Euclidean duals of TGRS codes and presented some implicit constructions of TGRS codes that are Euclidean self-dual.

The hull of a linear code $C$ is the intersection of $C$ with its dual $C^\perp$, where the dual is often defined with respect to  Euclidean inner product. The hull of a linear code $C$ was first studied by Assmus {\em et al.} \cite{AssKey} and was used to classify finite projective planes. Application of linear codes with small hull dimensions can be found in \cite{Leon82,Leon91,Sendrier00} for computational complexity and in \cite{CarGui} for increasing a security level against side 
channel attacks and fault injection attacks.

Recently, the hull dimension of a linear code has been applied to compute some parameters of an entanglement-assisted quantum error-correcting code (EAQECC) when such a quantum code is constructed from a classical linear code, using the methods proposed by Wilde {\em et al. }\cite{WilBru}. 
For the works on Euclidean hulls, the reader is referred to \cite{CarMesTanQiPel18,CarLiMes,LiZeng,QCM,Sok1D,Sok1D2,GGJT18,LCC}.
In \cite{GGJT18}, Guenda \emph{et al.} studied the so-called $\ell$-intersection pair of linear codes, and they completely determined the $q$-ary MDS EAQECCs of length $n \leq q+1$ for all possible parameters. 
It is quite difficult to use classical linear codes to construct new $q$-ary MDS EAQECCs with the code length $n > q+1$. However, by considering Hermitian hulls of MDS linear codes, the authors \cite{FangFuLiZhu,PerPel} could construct $q$-ary MDS EAQECCs for such lengths. It is has been known in \cite{FangFuLiZhu} that with respect to Hermitian inner product, it is very hard to construct linear codes with parameters $[n,k]_{q^2}$ for $k>\lfloor \frac{n+q-1}{q+1} \rfloor$ such that their Hermitian hull dimensions can be computed explicitly.

In this paper, we study Euclidean and Hermitian hulls of GRS codes and TGRS codes, as well as the Hermitian hulls of Roth-Lempel typed codes. We present explicit constructions of MDS and AMDS linear codes with arbitrary dimensional hulls for both Euclidean and Hermitian cases. Under some special condition, for instance $(n-1)|(q^2-1)$, we present a new method to construct, from Roth-Lempel typed codes, MDS $[n,k]$ and $[n+1,k]$ codes, as well as almost MDS $[n+2,k]$ codes over $\F_{q^2}$ with $k>\lfloor \frac{n+q-1}{q+1} \rfloor$ such that their Hermitian hull dimensions can be computed.
The key idea to obtain MDS linear codes with such large dimensions is to consider Hermitian self-orthogonal codes with parameters $[n,k]$ over $\F_{q^2}$.
We lengthen the Hermitian self-orthogonal codes, and under the special condition, the Hermitian hull dimension of the lengthened codes can be explicitly computed.
As a consequence, we obtain MDS linear codes with large dimension $k>\lfloor \frac{n+q-1}{q+1} \rfloor$, and these codes give rise to MDS EAQECCs with new parameters that can not be constructed in \cite{FangFuLiZhu}. As an application, we provide several classes of entanglement-assisted quantum error correcting codes with new parameters. 

The paper is organized as follows. Section \ref{section:pre} collects the notations and definitions of TRS codes and TGRS codes and studies the hulls of TGRS codes with respect to Euclidean and Hermitian inner products.
Section \ref{section:construction} presents explicit constructions of TGRS codes with arbitrary hull dimensions for both inner products. Roth-Lempel typed codes are also introduced and some method to compute their Hermitian hull dimensions is provided. 
Section \ref{section:application} gives an application of the codes, constructed in the previous sections, to entanglement assisted quantum error-correcting codes.

\section{Preliminaries}\label{section:pre}
\subsection{Hulls of linear codes}
A linear code of length $n$, dimension $k$, and minimum distance $d$ over the finite field $\F_q$ is denoted as $[n,k,d]_q$ code. 
A code is called {\em Maximum Distance Separable} ({MDS}) if its minimum distance $d=n-k+1$ and is called {\em almost} MDS (AMDS) if it is minimum distance $d=n-k$. 

The {\em Euclidean} (resp. {\em Hermitian}) inner product of ${\bf{x}}=(x_1,
\dots, x_n)$ and ${\bf{y}}=(y_1, \dots, y_n)$ in $\F_{q}^n$ (resp. $\F_{q^2}^n$) is defined by
$$<{\bf{x}},{\bf{y}}>_E=\sum_{i=1}^n x_i y_i, (\text{resp. }<{\bf{x}},{\bf{y}}>_H=\sum_{i=1}^n x_i y_i^q).
$$ 
The Euclidean (resp. Hermitian) {\em dual} of $C$,
denoted by $C^{\perp_E}$ (resp. $C^{\perp_H}$), is the set of vectors orthogonal to every
codeword of $C$ under the Euclidean (resp. Hermitian) inner product. A linear code $C$ is called Euclidean (resp. Hermitian) self-orthogonal if $C\subseteq C^{\perp_E}$ (resp. $C\subseteq C^{\perp_H}$). 
The Euclidean (resp. Hermitian) hull of a linear code $C$ is 
$$Hull_E(C):=C\cap C^{\perp_E}(\text{resp. }Hull_H(C):=C\cap C^{\perp_H}).$$ 
For a linear code $C\subseteq \F_q^n$ and ${\bf v}=(v_1,\hdots,v_n)\in (\F_q^*)^n$, we define

$$
{\bf v}\cdot C:=\{(v_1c_1,\hdots,v_nc_n)| {\bf c}=(c_1,\hdots,c_n)\in C\}.
$$
It is easy to see that ${\bf v}\cdot C$ is a linear code if and only if $C$ is a linear code. Moreover, both codes have the same dimension, minimum Hamming distance, and weight distribution.

\subsection{TRS and TGRS codes}

In this subsection, we review some notations and basic knowledge on twisted Reed-Solomon (TRS) codes \cite{BeePuRos17}. In particular,  we present twisted generalized Reed-Solomon (TGRS) codes and exhibit their check matrices.

Let $\mathcal P_k[x]$ be a $k$-dimensional $\mathbb{F}_q$-linear subspace of  $ \mathbb{F}_q[x]$. Let $\alpha_1,\ldots,\alpha_n $ be distinct elements in $\mathbb{F}_q$ and  $\alpha=(\alpha_1,\ldots,\alpha_n)$. Let $v_1,\ldots, v_n$ be non-zero elements in $\Bbb F_q$ and ${\bf v}=(v_1,\ldots, v_n)$. We call $\alpha_1,\ldots,\alpha_n$ the evaluation points. Define the evaluation map of $\alpha$ on $\mathcal P_k [x]$ by
$$\Phi_{\alpha}:\mathcal P_k[x]\longrightarrow \mathbb{F}_q^n, f(x)\longmapsto \Phi_\alpha(f(x))=(f(\alpha_1),\ldots, f(\alpha_n)).$$
For $k$, $t$, and $h$ being positive integers with  $0\leq h < k \leq q$ and  $\eta \in \Bbb F_q^*=\mathbb{F}_q\backslash\{0\}$, we define the set of $(k,t,h,\eta)$-twisted polynomials as
$$\mathcal P_k[x;t,h,\eta]=\{\sum_{i=0}^{k-1}a_ix^i+\eta a_hx^{k-1+t}: a_h,a_i\in \mathbb{F}_q, 0\le i\le k-1\},$$
which is a $k$-dimensional $\Bbb F_q$-linear subspace.  We call $h$ the hook and $t$ the twist.

For $\alpha_1,\ldots,\alpha_n $ being distinct elements in $\Bbb F_q$,  $\alpha=(\alpha_1,\ldots,\alpha_n)$, and $v_1,\ldots, v_n$ being non-zero elements in $\Bbb F_q$ and ${\bf v}=(v_1,\ldots, v_n)$, we define 
 the TRS code of length $n$ and dimension $k$ as
$$TRS_{k}(\alpha;t,h,\eta)=\Phi_{\alpha}(\mathcal P_k[x;t,h,\eta]),$$
and
the TGRS code of length $n$ and dimension $k$ as
$$TGRS_{k}(\alpha,{\bf v};t,h,\eta)={\bf v}\cdot \Phi_{\alpha}(\mathcal P_k[x;t,h,\eta]).$$

 In the rest of the paper, we assume that $t=1$ and $h=k-1$, and for $\alpha=(\alpha_1,\ldots,\alpha_n)$, we denote
 \begin{equation}
 s(\alpha)=\alpha_1+\hdots+\alpha_n\text{ and }  u_i=\prod\limits_{j=1,j\neq i}^n(\alpha_i-\alpha_j)^{-1}, 1\le i\le n.
 \end{equation}

 Let  $G_k$ be a  generator matrix  of  $TGRS_{k}(\alpha,{\bf v};1,k-1,\eta)$. Then

\begin{eqnarray}\label{eq:generator-matrix}
  G_k &=&\left(
  \begin{array}{cccc}
    v_1 & \ldots & v_n \\
    v_1\alpha_1 & \ldots & v_n\alpha_n \\
     \vdots &\vdots  & \vdots \\
    v_1(\alpha_1^{k-1}+\eta \alpha_1^k) & \ldots & v_n(\alpha_n^{k-1}+\eta \alpha_n^k) 
  \end{array}
\right).
\end{eqnarray}

 The parity check matrices $H_{n-k}$ of the $TGRS_{k}(\alpha,{\bf v};t,h,\eta)$, with generator matrix (\ref{eq:generator-matrix}), are determined by \cite{HuangYueNiuLi} and 
are presented on the top of the next page.

\begin{table*}[th]
\begin{enumerate}
\item  for $ s(\alpha)\ne 0$ and $\eta \neq - s(\alpha)^{-1}$,
{
\begin{eqnarray}\label{parity-check:1}
  H_{n-k} &=& \left(
  \begin{array}{cccc}
   \frac{ u_1}{v_1}  & \ldots & \frac{u_n}{v_n} \\
    \frac{u_1}{v_1}\alpha_1  & \ldots & \frac{u_n}{v_n}\alpha_n \\
    \vdots  & \vdots & \vdots \\
    \frac{u_1}{v_1}\alpha_1^{n-k-2}  & \ldots & \frac{u_u}{v_n}\alpha_1^{n-k-2} \\
   \frac{ u_1}{v_1}(\alpha_1^{n-k-1}- \frac{\eta }{1+ s(\alpha)\eta}\alpha_1^{n-k}) & \ldots &\frac{ u_n}{v_n}(\alpha_n^{n-k-1}-\frac{\eta }{1+ s(\alpha)\eta}\alpha_n^{n-k}) \\
  \end{array}
\right);
\end{eqnarray}
}
\item for $ s(\alpha)=0$ and $\eta \neq 0$,
{
\begin{eqnarray}\label{parity-check:2}
  H_{n-k} &=& \left(
  \begin{array}{cccc}
   \frac{ u_1}{v_1} & \ldots & \frac{u_n}{v_n} \\
    \frac{u_1}{v_1}\alpha_1 & \ldots & \frac{u_n}{v_n}\alpha_n \\
    \vdots & \vdots   & \vdots \\
    \frac{u_1}{v_1}\alpha_1^{n-k-2} & \ldots & \frac{u_u}{v_n}\alpha_1^{n-k-2} \\
   \frac{ u_1}{v_1}(\alpha_1^{n-k-1}- \eta \alpha_1^{n-k}) & \ldots &\frac{ u_n}{v_n}(\alpha_n^{n-k-1}-\eta \alpha_n^{n-k}) \\
  \end{array}
\right);
\end{eqnarray}
}
 \item  for $ s(\alpha)\neq 0$ and $\eta = - s(\alpha)^{-1}$,
 {
\begin{eqnarray}\label{parity-check:3}
  H_{n-k} &=& \left(
  \begin{array}{cccc}
   \frac{ u_1}{v_1} & \ldots & \frac{u_n}{v_n} \\
    \frac{u_1}{v_1}\alpha_1 & \ldots & \frac{u_n}{v_n}\alpha_n \\
    \vdots & \vdots &   \vdots \\
    \frac{u_1}{v_1}\alpha_1^{n-k-2} & \ldots & \frac{u_u}{v_n}\alpha_1^{n-k-2} \\
   \frac{ u_1}{v_1}\alpha_1^{n-k} & \ldots &\frac{ u_n}{v_n}\alpha_n^{n-k} \\
  \end{array}
\right).
\end{eqnarray}
}
\end{enumerate}

\end{table*}

From the exhibition of the parity check matrices (\ref{parity-check:1})-(\ref{parity-check:3}), we deduce the following result.
\begin{prop}\label{prop:E-dual} The Euclidean duals of $TGRS_{k}(\alpha,{\bf v};t,h,\eta)$ are given as follows:
\begin{enumerate}
\item  for $ s(\alpha)\ne 0$ and $\eta \neq - s(\alpha)^{-1}$, 
$$\begin{array}{ll}
TGRS_{k}(\alpha,{\bf v};1,k-1,\eta)^{\perp_E}&\\
~~~~~=TGRS_{n-k}(\alpha,\frac{\bf u}{\bf v};1,n-k-1,\frac{-\eta}{1+ s(\alpha)\eta})
\end{array}$$
\item for $ s(\alpha)= 0$ and $\eta \neq 0$, 
$$
\begin{array}{ll}
TGRS_{k}(\alpha,{\bf v};1,k-1,\eta)^{\perp_E}&\\
~~~~~=TGRS_{n-k}(\alpha,\frac{\bf u}{\bf v};1,n-k-1,-\eta)&
\end{array}
$$
\end{enumerate}
\end{prop}
As a consequence of Proposition \ref{prop:E-dual}, we immediately derive the following lemma.
\begin{lem} \label{lem2.1}
Let ${\bf c}=(v_{1}f(\alpha_{1}), v_{2}f(\alpha_{2}), \ldots, v_{n}f(\alpha_{n}))$  a codeword of $TGRS_{k}(\alpha,{\bf v};1,k-1,\eta)$. Then, 
\begin{enumerate}
\item for $ s(\alpha)\not=0,\eta\not=- s(\alpha)^{-1}$, ${\bf c}$ is contained in $TGRS_{k}(\alpha,{\bf v};1,k-1,\eta)^{\perp_E}$ if and only if there exists a polynomial $g(x)\in \mathcal P_{n-k}[x;1,n-k-1,\frac{-\eta}{1+ s(\alpha)\eta}]$ such that
\begin{equation*}
  \begin{split}
      & (v_{1}^{2}f(\alpha_{1}), v_{2}^{2}f(\alpha_{2}), \ldots, v_{n}^{2}f(\alpha_{n})) \\
       & =(u_{1}g(\alpha_{1}), u_{2}g(\alpha_{2}),\ldots, u_{n}g(\alpha_{n})).
   \end{split}
\end{equation*}
\item for $ s(\alpha)=0,\eta\not=0$, ${\bf c}$ is contained in $TGRS_{k}(\alpha,{\bf v};1,k-1,\eta)^{\perp_E}$ if and only if there exists a polynomial $g(x)\in \mathcal P_{n-k}[x;1,n-k-1,-\eta]$ such that
\begin{equation*}
  \begin{split}
      & (v_{1}^{2}f(\alpha_{1}), v_{2}^{2}f(\alpha_{2}), \ldots, v_{n}^{2}f(\alpha_{n})) \\
       & =(u_{1}g(\alpha_{1}), u_{2}g(\alpha_{2}),\ldots, u_{n}g(\alpha_{n})).
   \end{split}
\end{equation*}
\end{enumerate}
\end{lem}
The Hermitian case can also be obtained as follows.
\begin{lem}\label{lem2.2}
Let ${\bf c}=(v_{1}f(\alpha_{1}), v_{2}f(\alpha_{2}), \ldots, v_{n}f(\alpha_{n}))$  a codeword of $TGRS_{k}(\alpha,{\bf v};1,k-1,\eta)$. Then, 
\begin{enumerate}
\item for $ s(\alpha)\not=0,\eta\not=- s(\alpha)^{-1}$, ${\bf c}$ is contained in $TGRS_{k}(\alpha,{\bf v};1,k-1,\eta)^{\perp_H}$ if and only if there exists a polynomial $g(x)\in \mathcal P_{n-k}[x;1,n-k-1,\frac{-\eta}{1+ s(\alpha)\eta}]$ such that
\begin{equation*}
  \begin{split}
      & (v_{1}^{q+1}f(\alpha_{1}), v_{2}^{q+1}f(\alpha_{2}), \ldots, v_{n}^{q+1}f(\alpha_{n})) \\
       & =(u_{1}g(\alpha_{1}), u_{2}g(\alpha_{2}),\ldots, u_{n}g(\alpha_{n})).
   \end{split}
\end{equation*}
\item for $ s(\alpha)=0,\eta\not=0$, ${\bf c}$ is contained in $TGRS_{k}(\alpha,{\bf v};1,k-1,\eta)^{\perp_H}$ if and only if there exists a polynomial $g(x)\in \mathcal P_{n-k}[x;1,n-k-1,-\eta]$ such that
\begin{equation*}
  \begin{split}
      & (v_{1}^{q+1}f(\alpha_{1}), v_{2}^{q+1}f(\alpha_{2}), \ldots, v_{n}^{q+1}f(\alpha_{n})) \\
       & =(u_{1}g(\alpha_{1}), u_{2}g(\alpha_{2}),\ldots, u_{n}g(\alpha_{n})).
   \end{split}
\end{equation*}
\end{enumerate}
\end{lem}
\begin{proof} Assume that $ s(\alpha)\not=0,\eta\not=- s(\alpha)^{-1}$. Let ${\bf c}=(v_{1}f(\alpha_{1}), v_{2}f(\alpha_{2}), \ldots, v_{n}f(\alpha_{n}))$ be  a codeword of $TGRS_{k}(\alpha,{\bf v};1,k-1,\eta)$. Then,  ${\bf c}$ is contained in $TGRS_{k}(\alpha,{\bf v};1,k-1,\eta)^{\perp_H}$ if and only if ${\bf c}^q=(v_{1}^qf^q(\alpha_{1}), v_{2}^qf^q(\alpha_{2}), \ldots, v_{n}^qf^q(\alpha_{n}))$ is contained in $TGRS_{k}(\alpha,{\bf v};1,k-1,\eta)^{\perp_E}$, and from Proposition \ref{prop:E-dual}, ${\bf c}^q=(\frac{u_{1}}{v_1}g(\alpha_{1}), \frac{u_{2}}{v_2}g(\alpha_{2}),\ldots, \frac{u_{n}}{v_n}g(\alpha_{n}))$ for some $g(x)\in \mathcal P_{n-k}[x;1,n-k-1,-\eta]$. By equating ${\bf c}^q$ in both cases, the proof point 1) is completed. The proof of point 2) follows with the same reasoning.
\end{proof}

It is well known that both TRS and TGRS codes are either MDS or AMDS. In \cite{BeePuRos17}, the authors give necessary and sufficient conditions that TGRS are MDS and AMDS, and we rephrase the statement as follows:
\begin{lem}[\cite{BeePuRos17}] \label{lem:BeePuRon}
Let $\alpha_1,\ldots,\alpha_n$ be distinct elements in $\F_q$, $\alpha=(\alpha_1,\ldots,\alpha_n)$, 
${\bf v}=(v_1,\ldots, v_n)\in \left(\F_q^*\right)^n$, and $\eta\in \F_q^*.$ Denote
$$
S_k=\{\sum\limits_{i\in J}\alpha_i: J\subset \{1,\hdots, n\},\sharp J=k\}.
$$

Then,
\begin{enumerate}
\item the code $TGRS_{k}(\alpha,{\bf v};1,k-1,\eta)$ is MDS if and only if $\eta^{-1}\not=-s$ for any $s \in S_k$;
\item the code $TGRS_{k}(\alpha,{\bf v};1,k-1,\eta)$ is AMDS if and only if $\eta^{-1}=-s$ for some $s\in S_k$.
\end{enumerate}
\end{lem}

\section{Constructing hulls of MDS codes with non Reed-Solomon types}\label{section:construction}
\subsection{Constructing Euclidean hulls of TGRS codes}

In this subsection, we give, by applying Lemma \ref{lem2.1} and choosing a suitable set of evaluation points, an explicit construction of TGRS codes with arbitrary Euclidean hull dimensions for $q$ even.

\begin{thm}\label{thm:q:even}
Let $q$ be even and $(n-1)|(q-1)$. Then, for any $0\le r\le k-1$, 
there exists an $[n,k,\ge n-k]_q$ code  $C$ for $k\le \lfloor \frac{n}{2}\rfloor $ with $0\le \dim Hull_E(C)\le k-r$.
\end{thm}

\begin{proof} Denote $U=\{a\in \F_q: a^n=a\}$, label the elements of $U$ as $\alpha_{1}, \alpha_{2}, \ldots, \alpha_{n}$, and set $\alpha=(\alpha_{1}, \alpha_{2}, \ldots, \alpha_{n})$. It is easy to check that $s(\alpha)=\sum\limits_{i=1}^n\alpha_i=0.$
Choose $\eta\in \F_q^*$. Take $a \in \mathbb{F}_{q}^{*}$ with $a^{2} \neq 1$.

Since $q$ is even, each element of $\F_q$ is a square element in $\mathbb{F}_{q}$, and thus there exist $v_{1}, \ldots,  v_{n} \in \mathbb{F}_{q}^{*}$ such that
$ u_{i}= v_{i}^{2}$, for $1 \leq i \leq n$. Let $0\le r \le k-1$. Denote $\textbf{v}=(a v_{1},\ldots, a v_{r}, v_{r+1},\ldots, v_{n})$ and ${C}=TGRS_{k}(\alpha,{\bf v};1,k-1,\eta)$.
Let $\textbf{c}=(a v_{1}f(\alpha_{1}),\ldots, a v_{r}f(\alpha_{r}),v_{r+1}f(\alpha_{r+1}),\ldots,v_{n}f(\alpha_{n}))$ 
be a non-zero codeword in $Hull_{E}(C)$ with $\deg(f(x)) \leq k$. By Lemma \ref{lem2.1}, there exists a non-zero polynomial $g(x)\in \mathcal P_{n-k}[x;1,n-k-1,\frac{-\eta}{1+s(\alpha)\eta}]$ such that
\begin{eqnarray*}
   (a^{2} v_{1}^{2}f(\alpha_{1}),\ldots, a^{2} v_{r}^{2}f(\alpha_{r}),v_{r+1}^{2}f(\alpha_{r+1}),\ldots,v_{n}^{2}f(\alpha_{n})) \\
    =(u_{1}g(\alpha_{1}),\ldots, u_{r}g(\alpha_{r}),u_{r+1}g(\alpha_{r+1}),\ldots,u_{n}g(\alpha_{n})).
\end{eqnarray*}
Since $ u_{i}=v_{i}^{2}$, we have
\begin{equation}\label{eq:1}
  \begin{split}
     (a^{2}  u_{1}f(\alpha_{1}),\ldots, a^{2}  u_{r}f(\alpha_{r}),  u_{r+1}f(\alpha_{r+1}),\ldots, u_{n}f(\alpha_{n})) \\
      =(u_{1}g(\alpha_{1}),\ldots, u_{r}g(\alpha_{r}), u_{r+1}g(\alpha_{r+1}),\ldots, u_{n}g(\alpha_{n})).
  \end{split}
\end{equation}
By matching the last $n-r$ coordinates of Eq. (\ref{eq:1}), it follows that $ f(\alpha_{i})=g(\alpha_{i})$ for $i=r+1,\hdots, n$. Denote $h(x)=f(x)-g(x)$. Then, $h(x)$ has at least $n-r$ distinct roots. 
Since $2k\le n$, we get that $\deg(f(x)) \leq k \le n-k< n-r$. Moreover, since $\deg(g(x)) \leq n-k< n-r$,
it follows that $\deg (h)< n-r$, and thus $h(x)=f(x)-g(x)=0$. 
Since $g(x)\in \mathcal P_{n-k}[x;1,n-k-1,\frac{-\eta}{1+s(\alpha)\eta}]$ and $f(x)\in \mathcal P_{k}[x;1,k-1,\eta]$, by matching the polynomial coefficients in $f$ and $g$,
it is easy to obtain that $\deg (f(x))\le k-1.$

By considering the first $r$ coordinates of Eq. (\ref{eq:1}), we obtain that
\[a^{2}  u_{i}f(\alpha_{i})= u_{i}g(\alpha_{i})= u_{i}f(\alpha_{i}),\]
for any $1 \leq i \leq r$. It follows from $a^{2} \neq 1$ and $ u_{i} \neq 0$ that $f(\alpha_{i})=0$ for any $1\le i\le r$. Thus, $f(x)$ can be written as:
\[f(x)=f_1(x)\prod_{i=1}^{r}(x-\alpha_{i}),\]
for some $f_1(x) \in \mathbb{F}_{q}[x]$ with $\deg(f_1(x)) \leq k-r-1$.
It follows that $\dim(Hull_{E}(C)) \leq k-r$.

Next, we show that $\dim(Hull_{E}(C)) \geq\dim (C)=k\ge k-r$. Let $f(x)$ be a polynomial of form $f_1(x)\prod_{i=1}^{r}(x-\alpha_{i})$, where $f_1(x) \in \mathbb{F}_{q}[x]$ and $\deg(f_1(x)) \leq k-r-1$. Take $g(x)= f(x)$. Then, $\deg(g(x)) \leq n-k$ and
\begin{eqnarray*}
\begin{split}
   (a^{2} v_{1}^{2}f(\alpha_{1}),\ldots, a^{2} v_{r}^{2}f(\alpha_{r}),v_{r+1}^{2}f(\alpha_{r+1}),\ldots,v_{n}^{2}f(\alpha_{n}))\\
    =(u_{1}g(\alpha_{1}),\ldots, u_{r}g(\alpha_{r}),u_{r+1}g(\alpha_{r+1}),\ldots,u_{n}g(\alpha_{n})).
\end{split}
\end{eqnarray*}
By Lemma \ref{lem2.1}, the vector
$(a v_{1}f(\alpha_{1}),\ldots,a v_{r}f(\alpha_{r}),$ $v_{r+1}f(\alpha_{r+1}),\ldots,v_{n}f(\alpha_{n})) \in Hull_{E}(C).$
Therefore, $\dim(Hull_{E}(C)) \geq\dim (C)=k\ge k-r$, hence $\dim(Hull_{E}(C)) = k-r$.
\end{proof}
By applying Lemma \ref{lem:BeePuRon}, we deduce the following result.
\begin{cor}
Let $q$ be even and $(n-1)|(q-1)$. Put $U=\{a\in \F_q: a^n=a\}$ and $S_k=\{\sum\limits_{\alpha_i\in J}\alpha_i: J\subset U,\sharp J=k\}$. Then, for any $0\le r\le k-1$, 
\begin{enumerate}
\item there exists an MDS $[n,k, n-k+1]_q$ code  $C$ for $k\le \lfloor \frac{n}{2}\rfloor $ with $0\le \dim Hull_E(C)\le k-r$ if there exists $\eta\in\F_{q}^*$ such that $\eta^{-1}\not=-a,\forall a\in S_k$;
\item there exists an AMDS $[n,k, n-k]_q$ code  $C$ for $k\le \lfloor \frac{n}{2}\rfloor $ with $0\le \dim Hull_E(C)\le k-r$ if there exists $\eta\in\F_{q}^*$ such that $\eta^{-1}=-a,$ for some $a\in S_k$.
\end{enumerate}
\end{cor}

\begin{thm}\label{thm:q:odd}
Let $q=p^m$ be an odd prime power. Then, for any $n\le q-2$ and $0\le r\le k-1$, 
\begin{enumerate}
\item there exists an $[n=2k+1,k,\ge k+1]_{q^2}$ code  $C$ with $0\le \dim Hull_E(C)\le k-r-1$, and
\item there exists an $[n,k,\ge n-k]_{q^2}$ code  $C$ for $k\le \lfloor \frac{n}{2}\rfloor -1$ with $0\le \dim Hull_E(C)\le k-r$.
\end{enumerate}
\end{thm}

\begin{proof}
First note that $\sum\limits_{a\in \F_q}a=0$. Denote ${\bf v}=\F_q\backslash \{0,1\}$. Choose $U$ to be a subset of $V$ of size $n$ such that $\sum\limits_{a\in U}a\not=0$, and put $$P_U(x)=\prod\limits_{\alpha\in U}(x-\alpha).$$
Label the elements of $U$ as $\alpha_{1}, \alpha_{2}, \ldots, \alpha_{n}$, and set $\alpha=(\alpha_{1}, \alpha_{2}, \ldots, \alpha_{n})$ as well as $s(\alpha)=\sum\limits_{i=1}^n\alpha_i$. 
It is easy to check that:
$$
u_i=1/P_U'(\alpha_i).
$$
Choose $\eta=-2/s(\alpha)$. Take $a \in \mathbb{F}_{q^2}^{*}$ with $a^{2} \neq 1$.

Since $u_i$ is in $\F_q$, it is a square element in $\mathbb{F}_{q^2}$, and thus there exist $v_{1}, \ldots,  v_{n} \in \mathbb{F}_{q^2}^{*}$ such that
$ u_{i}= v_{i}^{2}$, for $1 \leq i \leq n$. Let $0\le r \le k-1$. Denote $\textbf{v}=(a v_{1},\ldots, a v_{r}, v_{r+1},\ldots, v_{n})$ and ${C}=GTRS_{k}(\alpha,{\bf v};1,k-1,\eta)$.
Let $\textbf{c}=(a v_{1}f(\alpha_{1}),\ldots, a v_{r}f(\alpha_{r}),v_{r+1}f(\alpha_{r+1}),\ldots,v_{n}f(\alpha_{n}))$ 
be a non-zero codeword in $Hull_{E}(C)$ with $\deg(f(x)) \leq k$. By Lemma \ref{lem2.1}, there exists a non-zero polynomial $g(x)\in \mathcal P_{n-k}[x;1,n-k-1,\frac{-\eta}{1+s(\alpha)\eta}]$ such that
\begin{eqnarray*}
    (a^{2} v_{1}^{2}f(\alpha_{1}),\ldots, a^{2} v_{r}^{2}f(\alpha_{r}),v_{r+1}^{2}f(\alpha_{r+1}),\ldots,v_{n}^{2}f(\alpha_{n})) \\
    =(u_{1}g(\alpha_{1}),\ldots, u_{r}g(\alpha_{r}),u_{r+1}g(\alpha_{r+1}),\ldots,u_{n}g(\alpha_{n})).
\end{eqnarray*}
Since $ u_{i}=v_{i}^{2}$, we have
\begin{equation}\label{eq:2}
  \begin{split}
     (a^{2}  u_{1}f(\alpha_{1}),\ldots, a^{2}  u_{r}f(\alpha_{r}),  u_{r+1}f(\alpha_{r+1}),\ldots, u_{n}f(\alpha_{n})) \\
      =(u_{1}g(\alpha_{1}),\ldots, u_{r}g(\alpha_{r}), u_{r+1}g(\alpha_{r+1}),\ldots, u_{n}g(\alpha_{n})).
  \end{split}
\end{equation}
By matching the last $n-r$ coordinates of Eq. (\ref{eq:2}), it follows that $ f(\alpha_{i})=g(\alpha_{i})$ for $i=r+1,\hdots, n$. Denote $h(x)=f(x)-g(x)$. Then, $h(x)$ has at least $n-r$ distinct roots. 
\begin{enumerate}[1)]
\item Case $n-1=2k:$\\
If $n-1=2k$, then we get that $\deg(f(x)) \leq k = n-k-1< n-r-1$. Moreover, since $\deg(g(x)) \leq n-k< n-r$. 
It follows that $\deg (h)< n-r$, and thus $h(x)=f(x)-g(x)=0$. 
Since $g(x)\in \mathcal P_{n-k}[x;1,n-k-1,\frac{-\eta}{1+s(\alpha)\eta}]$ and $f(x)\in \mathcal P_{k}[x;1,k-1,\eta]$, 
it is easy to obtain that $\deg (f(x))\le k-2.$

By considering the first $r$ coordinates of Eq. (\ref{eq:2}), we obtain that
\[a^{2}  u_{i}f(\alpha_{i})= u_{i}g(\alpha_{i})= u_{i}f(\alpha_{i}),\]
for any $1 \leq i \leq r$. It follows from $a^{2} \neq 1$ and $ u_{i} \neq 0$ that $f(\alpha_{i})=0$ for any $1\le i\le r$. Thus, $f(x)$ can be written as:
\[f(x)=f_1(x)\prod_{i=1}^{r}(x-\alpha_{i}),\]
for some $f_1(x) \in \mathbb{F}_{q}[x]$ with $\deg(f_1(x)) \leq k-r-2$.
It follows that $\dim(Hull_{E}(C)) \leq k-r-1$.

Next, we show that $\dim(Hull_{E}(C)) \geq\dim (C)=k\ge k-r-1$.  Let $f(x)$ be a polynomial of form $f_1(x)\prod_{i=1}^{r}(x-\alpha_{i})$, where $f_1(x) \in \mathbb{F}_{q}[x]$ and $\deg(f_1(x)) \leq k-r-2$. Take $g(x)= f(x)$. Then, $\deg(g(x)) \leq n-k$ and
\begin{eqnarray*}
\begin{split}
   (a^{2} v_{1}^{2}f(\alpha_{1}),\ldots, a^{2} v_{r}^{2}f(\alpha_{r}),v_{r+1}^{2}f(\alpha_{r+1}),\ldots,v_{n}^{2}f(\alpha_{n}))\\
    =(u_{1}g(\alpha_{1}),\ldots, u_{r}g(\alpha_{r}),u_{r+1}g(\alpha_{r+1}),\ldots,u_{n}g(\alpha_{n})).
\end{split}
\end{eqnarray*}
By Lemma \ref{lem2.1}, the vector
$(a v_{1}f(\alpha_{1}),\ldots,a v_{r}f(\alpha_{r}),$ $v_{r+1}f(\alpha_{r+1}),\ldots,v_{n}f(\alpha_{n})) \in Hull_{E}(C).$
Therefore, $\dim(Hull_{E}(C)) \geq\dim (C)=k\ge k-r-1$, hence $\dim(Hull_{E}(C)) = k-r-1$.

\item Case $n-1\not=2k$:\\
If $n-1\not=2k$, then we obtain that $\deg(f(x)) \leq k<n-k< n-r$.  Moreover, since $\deg(g(x)) \leq n-k< n-r$, it follows that $\deg (h)< n-r$, and thus $h(x)=f(x)-g(x)=0$. Since $g(x)\in \mathcal P_{n-k}[x;1,n-k-1,\frac{-\eta}{1+s(\alpha)\eta}]$ and $f(x)\in \mathcal P_{k}[x;1,k-1,\eta]$, 
it is easy to obtain that $\deg (f(x))\le k-1.$ With a similar discussion as point 1), we obtain that $\dim(Hull_{E}(C)) = k-r$.
\end{enumerate}
\end{proof}
\begin{exam} Take $q=13$, $\theta$ a primitive element of $\F_{13^2}$, $n=11$, $r=1$, $a=\theta$, and $\alpha=(2, 4, 8, 3, 6, 12, 11, 9, 5, 10, 7 )$. 
It follows that 
${\bf v}=(\theta^{133}, \theta^{161}, \theta^{119}, \theta^{119},     1,     7, \theta^{133},     1 ,   10, \theta^{161},    10)$, $s(\alpha)=12$, and $\eta=-2/s(\alpha)=2$. 

If we take $k=5$, then we get a $[11, 5, 6]_{13^2}$ code $C_5$ with $\dim Hull_E(C_5)=k-r-1=3.$ Its generator matrix is given as follows:
{\scriptsize
$$
\left(
\begin{array}{lllllllllll}
\theta^{134}&\theta^{161}&\theta^{119}&\theta^{119}&1&7&\theta^{133}&1&10&\theta^{161}&10\\
\theta^{50}&\theta^{91}&\theta^{63}&\theta^{77}&10&10&\theta^{147}&4&2&\theta^{49}&8\\
\theta^{134}&\theta^{21}&\theta^{7}&\theta^{35}&9&5&\theta^{161}&3&3&\theta^{105}&9\\
\theta^{50}&\theta^{119}&\theta^{119}&\theta^{161}&12&9&\theta^{7}&12&11&\theta^{161}&2\\
\theta^{50}&\theta^{21}&\theta^{133}&\theta^{49}&11&9&\theta^{147}&3&1&\theta^{35}&0\\
\end{array}
\right).
$$
}

If we take $k=4$, then we get a $[11, 4, 7]_{13^2}$ code $C_4$ with $\dim Hull_E(C_4)=k-r=3.$
Its generator matrix is given as follows:
{\scriptsize
$$
\left(
\begin{array}{lllllllllll}
\theta^{134}&\theta^{161}&\theta^{119}&\theta^{119}&1&7&\theta^{133}&1&10&\theta^{161}&10\\
\theta^{50}&\theta^{91}&\theta^{63}&\theta^{77}&10&10&\theta^{147}&4&2&\theta^{49}&8\\
\theta^{134}&\theta^{21}&\theta^{7}&\theta^{35}&9&5&\theta^{161}&3&3&\theta^{105}&9\\
\theta^{134}&\theta^{91}&\theta^{21}&\theta^{91}&5&5&\theta^{133}&4&5&\theta^{147}&0\\
\end{array}
\right).
$$
}
\end{exam}

\subsection{Constructing Hermitian hulls of TGRS codes}

In this subsection, we give, by applying Lemma \ref{lem2.2} and choosing suitable sets of evaluation points, explicit constructions of TGRS codes with arbitrary Hermitian hull dimensions.

\begin{thm}\label{thm:TGRS-hermitian} Let $q=p^m$ be a prime power and $1\le k\le \lfloor \frac{n}{q+1} \rfloor$. Assume that one of the following conditions holds:
\begin{enumerate}
\item $(n-1)|(q^2-1)$;
\item $n=tq$, $1\le t\le q-1$;
\item $n=(t+1)N+1$, $N|(q^2-1)$, $n_2=\frac{N}{\gcd (N,q+1)}$, $1\le t\le \frac{q-1}{n_2}-2$.
\end{enumerate}
Then, for any $0\le r\le k-1$, 
there exists an $[n,k, \ge n-k]_{q^2}$ code  $C$ with $0\le \dim Hull_H(C)\le k-r$.
\end{thm}

\begin{proof}
\begin{enumerate}
\item Denote $U=\{\alpha\in \F_{q^2}: \alpha^n=\alpha\}$, and put $$P_U(x)=\prod\limits_{\alpha\in U}(x-\alpha).$$
Label the elements of $U$ as $\alpha_{1}, \alpha_{2}, \ldots, \alpha_{n}$ with $\alpha_n=0$, and set $\alpha=(\alpha_{1}, \alpha_{2}, \ldots, \alpha_{n})$. It follows that $s(\alpha)=\sum\limits_{a\in U}a=0$. The derivative of $P_U(x)$ is $P_U'(x)=nx^{n-1}-1.$
We have that $P_U'(\alpha_i)=n-1$ for $1\le i\le n-1$ and $P_U'(\alpha_n)=-1$, which are all in $\F_{q}^*$, and thus 
$(P_U'(\alpha_i))_{1\le i\le n}$ can be written as a $(q+1)$-power element of $\F_{q^2}^*$.
Take $v_i=\frac{1}{\sqrt[q+1]{n-1}}$ for $1\le i\le n-1$ and $v_n=\frac{1}{\sqrt[q+1]{-1}}$.
Choose $\eta\in \F_{q^2}^*$ such that $\eta^q=-\eta$. Take $a \in \mathbb{F}_{q}^{*}$ with $a^{q+1} \neq 1$.
Let $0\le r \le k-1$. Denote $\textbf{v}=(a v_{1},\ldots, a v_{r}, v_{r+1},\ldots, v_{n})$ and ${C}=TGRS_{k}(\alpha,{\bf v};1,k-1,\eta)$.
Let $\textbf{c}=(a v_{1}f(\alpha_{1}),\ldots, a v_{r}f(\alpha_{r}),v_{r+1}f(\alpha_{r+1}),\ldots,v_{n}f(\alpha_{n}))$ 
be a non-zero codeword in $Hull_{H}(C)$ with $\deg(f(x)) \leq k$. By Lemma \ref{lem2.2}, there exists a non-zero polynomial $g(x)\in \mathcal P_{n-k}[x;1,n-k-1,\frac{-\eta}{1+s(\alpha)\eta}]$ such that
{
\begin{eqnarray*}
     (a^{q_1} v_{1}^{q_1}f^q(\alpha_{1}),\ldots, a^{q_1} v_{r}^{q_1}f^q(\alpha_{r}),v_{r+1}^{q_1}f^q(\alpha_{r+1}),\ldots,v_{n}^{q_1}f^q(\alpha_{n})) \\
    =(u_{1}g(\alpha_{1}),\ldots, u_{r}g(\alpha_{r}),u_{r+1}g(\alpha_{r+1}),\ldots,u_{n}g(\alpha_{n})),
\end{eqnarray*}
}
where $q_1=q+1$.

Since $ u_{i}=v_{i}^{q+1}$, we have
{
\begin{equation}\label{3}
  \begin{split}
     (a^{q_1}  u_{1}f^q(\alpha_{1}),\ldots, a^{q_1}  u_{r}f^q(\alpha_{r}),  u_{r+1}f^q(\alpha_{r+1}),\ldots, u_{n}f^q(\alpha_{n})) \\
      =(u_{1}g(\alpha_{1}),\ldots, u_{r}g(\alpha_{r}), u_{r+1}g(\alpha_{r+1}),\ldots, u_{n}g(\alpha_{n})).
  \end{split}
\end{equation}
}
By matching the last $n-r$ coordinates of Eq. \eqref{3}, it follows that $ f^q(\alpha_{i})=g(\alpha_{i})$ for $i=r+1,\hdots, n$. Denote $h(x)=f^q(x)-g(x)$. Then, $h(x)$ has at least $n-r$ distinct roots. 
Since $k(q+1)\le n $, it follows that $\deg(f^q(x)) \leq qk\le n-k< n-r$.  Moreover, since $\deg(g(x)) \leq n-k< n-r$, it follows that $\deg (h)< n-r$, and thus $h(x)=f^q(x)-g(x)=0$. 

Since $g(x)\in \mathcal P_{n-k}[x;1,n-k-1,\frac{-\eta}{1+s(\alpha)\eta}]$ and $f(x)\in \mathcal P_{k}[x;1,k-1,\eta]$, 
it is easy to obtain that $\deg (f(x))\le k-1.$

By considering the first $r$ coordinates of Eq. \eqref{3}, we obtain that
\[a^{q+1}  u_{i}f(\alpha_{i})= u_{i}g(\alpha_{i})= u_{i}f(\alpha_{i}),\]
for any $1 \leq i \leq r$. It follows from $a^{q+1} \neq 1$ and $ u_{i} \neq 0$ that $f(\alpha_{i})=0$ for any $1\le i\le r$. Thus, $f(x)$ can be written as:
\[f(x)=f_1(x)\prod_{i=1}^{r}(x-\alpha_{i}),\]
for some $f_1(x) \in \mathbb{F}_{q}[x]$ with $\deg(f_1(x)) \leq k-r-1$.
It follows that $\dim(Hull_{H}(C)) \leq k-r$.

Next, we show that $\dim(Hull_{E}(C)) \geq\dim (C)=k\ge k-r$. Let $f(x)$ be a polynomial of form $f_1(x)\prod_{i=1}^{r}(x-\alpha_{i})$, where $f_1(x) \in \mathbb{F}_{q}[x]$ and $\deg(f_1(x)) \leq k-r-1$. Take $g(x)= f(x)$. Then, $\deg(g(x)) \leq n-k$ and
{
\begin{eqnarray*}
\begin{split}
   & (a^{q_1} v_{1}^{q_1}f^q(\alpha_{1}),\ldots, a^{q_1} v_{r}^{q_1}f^q(\alpha_{r}),v_{r+1}^{q_1}f^q(\alpha_{r+1}),\ldots,v_{n}^{q_1}f^q(\alpha_{n}))\\
    & =(u_{1}g(\alpha_{1}),\ldots, u_{r}g(\alpha_{r}),u_{r+1}g(\alpha_{r+1}),\ldots,u_{n}g(\alpha_{n})).
\end{split}
\end{eqnarray*}
}
By Lemma \ref{lem2.1}, the vector
$(a v_{1}f(\alpha_{1}),\ldots,a v_{r}f(\alpha_{r}),$ $v_{r+1}f(\alpha_{r+1}),\ldots,v_{n}f(\alpha_{n})) \in Hull_{H}(C).$
Therefore, $\dim(Hull_{H}(C)) \geq\dim (C)=k\ge k-r$, hence $\dim(Hull_{H}(C)) = k-r$.
\item Fix an element $\beta\in \F_{q^2}\backslash \F_q$, and label the elements of $\F_q$ as $\{a_1,\hdots, a_q\}$. 
Denote $\beta_{i,j}=a_i\beta+a_j$ for $1\le i\le t$ and $1\le j\le q$. Put $U=\{\beta_{i,j}|1\le i\le t, 1\le j\le q\}$, set $\alpha=(\beta_{i,j})_{1\le i\le t, 1\le j\le q}$,
and write
\begin{equation*}
P_U(x)=\prod\limits_{
1\le i\le t,1\le j\le q
}(x-\beta_{i,j}).
\end{equation*}
It follows that $s(\alpha)=\sum\limits_{a\in U}a=0$.
The derivative $P_U'(x)$ at $\beta_{i_0,j_0}\in U$ is computed as follows:
{
\begin{equation*}
\begin{array}{ll}
P_U'(\beta_{i_0,j_0})
&=
\prod\limits_{
\begin{array}{c}
1\le i\le t,1\le j\le q\\
(i,j)\not= (i_0,j_0)\\
\end{array}
}(\beta_{i_0,j_0}-\beta_{i,j})\\
&=
\prod\limits_{
\begin{array}{c}
1\le j\le q\\
j\not= j_0\\
\end{array}
}(a_{i_0}\beta+a_{j_0}-a_{i_0}\beta-a_j)\\
&\prod\limits_{
\begin{array}{c}
1\le i\le t,\\
1\le j\le q\\
i\not= i_0\\
\end{array}
}(a_{i_0}\beta+a_{j_0}-a_{i}\beta-a_j)\\
&=
\prod\limits_{
\begin{array}{c}
1\le j\le q\\
j\not= j_0\\
\end{array}
}(a_{j_0}-a_j)\\
&\prod\limits_{
\begin{array}{c}
1\le i\le t\\
1\le j\le q\\
i\not= i_0\\
\end{array}
}((a_{i_0}-a_i)\beta+(a_{j_0}-a_j))\\
&=
-
\prod\limits_{
\begin{array}{c}
1\le i\le t\\
i\not= i_0\\
\end{array}
}((a_{i_0}-a_i)^q\beta^q+(a_{j_0}-a_j)\beta)\\
&=
-(\beta^q-\beta)^{t-1}
\prod\limits_{
\begin{array}{c}
1\le i\le t\\
i\not= i_0\\
\end{array}
}(a_{i_0}-a_i).\\
\end{array}
\end{equation*}
}
The two last equalities hold due to the fact that the product of all element in $\F_q^*$ is equal to $-1$.

It follows that $P_U'(\beta_{i_0,j_0})\in \F_{q}^*$ and thus can be written as a $(q+1)$-power element of $\F_{q^2}^*$.
Set $(v_1,\hdots,v_n)=\left(\frac{1}{\sqrt[q+1]{P_U'(\beta_{i_0,j_0})}}\right)_{1\le i_0\le t,1\le j_0\le q}$.
Choose $\eta\in \F_{q^2}^*$ such that $\eta^q=-\eta$. Take $a \in \mathbb{F}_{q}^{*}$ with $a^{q+1} \neq 1$.
Let $0\le r \le k-1$. Denote $\textbf{v}=(a v_{1},\ldots, a v_{r}, v_{r+1},\ldots, v_{n})$ and ${C}=TGRS_{k}(\alpha,{\bf v};1,k-1,\eta)$. 
The rest follows with the same reasoning as that in the first part.
\item Assume that $N|(q^2-1)$, and put
\begin{equation}
n_{1}=\gcd (N, q+1), 
\label{eq:n1}
\end{equation} and
\begin{equation}
n_{2}=\frac{N}{\gcd (N,q+1)}. 
\label{eqn2}
\end{equation}
From (\ref{eq:n1}), it follows that $\gcd (n_2,\frac{q+1}{n_1})=1$. We get that $n_2|(q-1)\frac{q+1}{n_1}$ since $N|(q^2-1)$, and thus $n_2|(q-1)$.
Let $U_N$ and $V_N$  be two multiplicative subgroups of $\mathbb{F}_{q^{2}}^{*}$ generated by $\theta^{\frac{q^{2}-1}{N}}$ and $\theta^{\frac{q+1}{n_{1}}}$, respectively, where $\theta$ is a primitive element of $\mathbb{F}_{q^{2}}$. We can easily check that their orders are $\sharp U_N=N$ and $\sharp V_N=(q-1)n_{1}$. Since $\frac{q^{2}-1}{N}=\frac{q+1}{n_{1}}\cdot\frac{q-1}{n_{2}}$, we obtain that $ \frac{q+1}{n_{1}} \mid \frac{q^{2}-1}{N}$, and thus $U_N$ is a subgroup of $V_N$.
Let $\alpha_{1}U_N, \ldots , \alpha_{\frac{q-1}{n_{2}}-1}U_N $ be all the distinct cosets of $V_N$ different from $U_N$. 

For $1 \leq t \leq \frac{q-1}{n_{2}}-2$, label the elements of $U_N$ as $a_1,\hdots,a_N$. Put 
\begin{equation}
U=U_N \bigcup^{t}\limits_{j=1}\alpha_{j}U_N\cup \{0\},
\label{eq:multi-coset}
\end{equation}
label the elements of $U$ as $a_1,\hdots,a_{(t+1)N+1}$, set $\alpha=(a_1,\hdots,a_{(t+1)N+1})$, and
write $$P_U(x)=\prod\limits_{\alpha\in U}(x-\alpha).$$
It follows that $s(\alpha)=\sum\limits_{a\in U}a=0.$
The derivative $P_U'(x)$ is computed as follows:
\begin{equation*}
\begin{array}{ll}
P_U'(x)=&((N+1)x^N-1)\prod\limits_{i=1}^t(x^N-\alpha_i^N)\\
&+Nx^N(x^N-1)\left(\sum\limits_{i=1}^t\prod\limits_{j=1,j\not=i}^t(x^N-\alpha_j^N)\right).\\
\end{array}
\end{equation*}

For $1\le j\le t,1\le s\le N$, we have 
\begin{equation} P_U'(\alpha_ju_s)=N\alpha_j^N(1-\alpha_j^N)\prod\limits_{i=1,i\not=j}^t(\alpha_j^N-\alpha_i^N).
\label{eq:derivative}
\end{equation}

From the fact that $\alpha_{j}$ is in $V_N$, we have that $\alpha_{j}=\theta^{e_j\frac{q+1}{n_{1}}}$ for some positive integer $e_j$. Hence,
$\alpha_{j}^{N}=\theta^{e_jN\frac{q+1}{n_{1}}}=\theta^{e_jn_{2}(q+1)}$, and this shows that  $\alpha_{j}^{N}$ is an element of $\mathbb{F}^{*}_{q}$. So, we deduce that for any $1\le i\le n=(t+1)N+1$, we have $P_U'(a_i)\in \F_q^*$ and thus $P_U'(a_i)=\beta_i^{q+1}$ for some $\beta_i \in \F_{q^2}$.
Set $v_i=\frac{1}{\beta_i}$ for $1\le i\le n.$
Choose $\eta\in \F_{q^2}^*$ such that $\eta^q=-\eta$. Take $a \in \mathbb{F}_{q}^{*}$ with $a^{q+1} \neq 1$.
Let $0\le r \le k-1$. Denote $\textbf{v}=(a v_{1},\ldots, a v_{r}, v_{r+1},\ldots, v_{n})$ and ${C}=TGRS_{k}(\alpha,{\bf v};1,k-1,\eta)$.
The rest follows with the same reasoning as that in the first part.
\end{enumerate}
\end{proof}

By applying Lemma \ref{lem:BeePuRon}, we deduce the following result.
\begin{cor} Let $q=p^m$ be a prime power and $1\le k\le \lfloor \frac{n}{q+1} \rfloor$. 
Assume that one of the following conditions holds:
\begin{enumerate}
\item $(n-1)|(q^2-1)$ and $U=\{a\in \F_q: a^n=a\}$;
\item $n=tq$, $1\le t\le q-1$ and $U=\{a_i\beta+a_j|1\le i\le t, 1\le j\le q, a_i,a_j\in \F_q\}$, where $\beta\in \F_{q^2}\backslash \F_q$;
\item $n=(t+1)N+1$, $N|(q^2-1)$, $n_2=\frac{N}{\gcd (N,q+1)}$, $1\le t\le \frac{q-1}{n_2}-2$ and $U$ is given by (\ref{eq:multi-coset}).
\end{enumerate}
Put $S_k=\{\sum\limits_{\alpha_i\in J}\alpha_i: J\subset U,\sharp J=k\}$.
Then, for any $0\le r\le k-1$, 
\begin{enumerate}
\item there exists an MDS $[n,k, n-k+1]_{q^2}$ code  $C$  with $0\le \dim Hull_H(C)\le k-r$ 
if there exists $\eta\in\F_{q}^*$ such that $\eta^q=-\eta$ and $\eta^{-1}\not=-a,\forall a\in S_k$;
 \item there exists an AMDS $[n,k, n-k]_{q^2}$ code  $C$ with $0\le \dim Hull_H(C)\le k-r$ 
if if there exists $\eta\in\F_{q}^*$ such that $\eta^q=-\eta$ and $\eta^{-1}=-a,$ for some $a\in S_k$.
\end{enumerate}
\end{cor}

\subsection{Constructing Hermitian hulls of Roth-Lempel typed codes}
Recall that the Roth-Lempel typed code is an $[n+2,k]$ code which is either MDS or AMDS. The generator matrix of such a code has the following form:
\begin{equation}
G=(L:R),
\end{equation}
where 
$$
L=\left(
\begin{array}{cccccccc}
v_1&\hdots&v_{n}\\
v_1\alpha_1&\hdots&v_{n}\alpha_{n}\\
\vdots&\vdots&\vdots\\
v_1\alpha_1^{k-3}&\hdots&v_{n}\alpha_{n}^{k-3}\\
v_1\alpha_1^{k-2}&\hdots&v_{n}\alpha_{n}^{k-2}\\
v_1\alpha_1^{k-1}&\hdots&v_{n}\alpha_{n}^{k-1}\\
\end{array}
\right),
R=\left(
\begin{array}{cc}
0&0\\
0&0\\
\vdots&\vdots\\
0&0\\
1&0\\
\delta&1\\
\end{array}
\right).
$$
It should be noted that for $k\le \lfloor \frac{n-1+q}{q+1}\rfloor $, MDS $[n,k]_{q^2}$ codes, with Hermitian hull dimensions ranging between $1$ and $k$, have been constructed in\cite{FangFuLiZhu}. So for the rest of the paper, we consider $[n,k']_{q^2}$ codes with $k'>\lfloor \frac{n-1+q}{q+1}\rfloor$.
To be easier to deal with the Hermitian hull, we consider codes with a slightly different form from the Roth-Lempel typed codes, that is, by slightly changing the entries in the matrix $R$. More precisely, we consider the codes with the following generator matrix:
\begin{equation}\label{eq:gen-RL}
G_{k+2}=(L_{k+2}:R_{k+2}),
\end{equation}
where $k\le \lfloor \frac{n-1+q}{q+1}\rfloor $ and
$$
L_{k+2}=\left(
\begin{array}{cccccccc}
v_1&\hdots&v_{n}\\
v_1\alpha_1&\hdots&v_{n}\alpha_{n}\\
\vdots&\vdots&\vdots\\
v_1\alpha_1^{k-1}&\hdots&v_{n}\alpha_{n}^{k-1}\\
v_1\alpha_1^{k}&\hdots&v_{n}\alpha_{n}^{k}\\
v_1\alpha_1^{k+1}&\hdots&v_{n}\alpha_{n}^{k+1}\\
\end{array}
\right),
R_{k+2}=\left(
\begin{array}{cc}
0&0\\
0&0\\
\vdots&\vdots\\
0&0\\
\lambda_1&0\\
\delta&\lambda_2\\
\end{array}
\right),
$$

We denote the conjugate transpose of a matrix $M$ by $M^{\dag}$, that is, if $M=m_{i,j}$ Then, $M^{\dag}=m_{j,i}^q.$ 
The following lemma enables us to easily compute the Hermitian hull dimension of a linear code if its generator matrix or parity check matrix is given.
\begin{lem}\textnormal{\cite{GJG18}}
\label{lem:hull-H}
Let $C$ be an $[n,k,d]_{q^2}$ code with parity check matrix $P$ and generator matrix $G$.
Then, $\textnormal{rank}(PP^{\dag})$ and $\textnormal{rank}(GG^{\dag})$ are independent of $P$ and $G$ so that
$$
\begin{array}{ll}
\textnormal{rank}(PP^{\dag})&=n-k-\dim(Hull_H(C)) \\
&= n-k-\dim(Hull_H(C^{\perp_H})),
\end{array}
$$
and
$$
\begin{array}{ll}
\textnormal{rank}(GG^{\dag})&=k-\dim(Hull_H(C)) \\
&= k-\dim(Hull_H(C^{\perp_H})).
\end{array}
$$
\end{lem}
The following lemma is useful for constructing an GRS code that is either Euclidean self-orthogonal \cite{JinXin} or Hermitian self-orthogonal \cite{FangFu,SokQSC}.
\begin{lem}\label{lem:existence} Put $\alpha=(\alpha_1,\hdots, \alpha_n)$ and ${\bf v}=(v_1,\hdots,v_n)$. 
\begin{enumerate}
\item If $u_i=v_i^2$ for $1\le i\le n$ and $k\le n/2$, then the code $GRS_{k}(\alpha,{\bf v})$ is Euclidean self-orthogonal over $\F_{q}$.
\item If $u_i=v_i^{q+1}$ for $1\le i\le n$ and $k\le \lfloor \frac{n-1+q}{q+1}\rfloor $, then the code $GRS_{k}(\alpha,{\bf v})$ is Hermitian self-orthogonal over $\F_{q^2}$.
\end{enumerate}
\end{lem}
We are now ready to provide existence of an (MDS) $[n,k]$ code $C$ with $k> \lfloor \frac{n+q-1}{q+1} \rfloor$ such that its Hermitian hull dimension $\dim(Hull_H(C))$ can be computed explicitly.
\begin{thm}\label{thm:embedding-new}Assume that $(n-1)|(q^2-1)$ and $1\le k\le \lfloor \frac{n+q-1}{q+1} \rfloor$. 
Then, there exists a linear code $C$ with parameters
 $[n+2,k+2,\ge n-k]_{q^2}$ such that $\dim(Hull_H(C))=k+2-i$ for $i=1,2$ if $(n-1)|(k+1)(q+1)$ or $(n-1)|k(q+1)$.
\end{thm}

\begin{proof} Take $U_{n-1}=\{ \alpha\in \F_{q^2}|\alpha^{n-1}=1\}$, and label the elements of $U_{n-1}$ as $\alpha_1,\hdots,\alpha_{n-1}$, as well as set $\alpha_n=0$. From Lemma \ref{lem:existence} and under the assumption in the theorem, the code $GRS_{k}(\alpha,{\bf v})$ is Hermitian self-orthogonal with parameters $[n,k]_{q^2}$, where $\alpha=(\alpha_1,\hdots, \alpha_n)$ and ${\bf v}=(v_1,\hdots,v_n$ with $u_j=v_j^{q+1}$ for $1\le j\le n$. 
Consider the code $C_{k+2}$ with generator matrix $G_{k+2}$ written as in (\ref{eq:gen-RL}) such that $\delta=0$. 
For $i=1,2$, we want to show that there exist $\lambda_1,\lambda_2\in \F_{q^2}^*$ such that $\text{rank}(G_{k+2}G_{k+2}^\dag)=i$.

Let $g_i$ be the $i$-th row of $L_{k+2}$. Consider the two following cases:\\

\noindent
{\bf Case $(n-1)|(k+1)(q+1)$:}
\begin{enumerate}[(a)]
\item It is easy to check that $<g_{i},g_{k+1}>_H=<g_{i},g_{k+2}>_H=0$ for $1\le i\le k$. 
\item It follows after taking the exponent of $\alpha_i$ modulo $n-1$ that $<g_{k+1},g_{k+1}>_H=\sum\limits_{l=1}^nv_l^{q+1}\alpha_l^{j}=0$, where $0< j\le n-2$. 
\item Under the condition $(n-1)|(k+1)(q+1)$, we have that $(\alpha_1^{(k+1)(q+1)},\hdots, \alpha_n^{(k+1)(q+1)})=(1,\hdots,1,0)$, and so 
$<g_{k+2},g_{k+2}>_H=v_1^{q+1}+\hdots+ v_{n-1}^{q+1}=-v_{n}^{q+1}\not=0$.
\item We now check that $<g_{k+1},g_{k+2}>_H=0$, that is, 
\begin{equation}
\sum\limits_{i=1}^n\alpha_i^{kq+k+1}v_i^{q+1}=0.
\label{eq:gkandgk1}
\end{equation}
\begin{itemize}
\item If $k(q+1)+1\le n-1$, then $k(q+1)+1\le n-2$, and thus $<g_{k+1},g_{k+2}>_H=0$.
\item  If $k(q+1)+1 > n-1$, then we can write $k(q+1)+1=(n-1)A+B$ with $0\le B<(n-1)$ and thus $<g_{k+1},g_{k+2}>_H=0$.
\end{itemize}
\end{enumerate}
\begin{itemize}
\item By choosing $\lambda_1\not=0$ and $\lambda_2\not=\pm v_n$, we obtain that 
$$
\begin{array}{ll}
G_{k+2}G_{k+2}^\dag&=L_{k+2}L_{k+2}^\dag+R_{k+2}R_{k+2}^\dag\\
&=\text{diag} (0,\hdots,0,\lambda_1^{q+1}, \lambda_2^{q+1}-v_n^{q+1}).
\end{array}
$$
It follows that $\dim(Hull_H(C))=k+2-\text{rank}(G_{k+2}G_{k+2}^\dag)=k$.
\item By choosing $\lambda_1\not=0$ and $\lambda_2=\pm v_n$, we obtain that $\dim(Hull_H(C))=k+2-\text{rank}(G_{k+2}G_{k+2}^\dag)=k+1$.
\end{itemize}
\noindent
{\bf Case $(n-1)|k(q+1)$:}
\begin{enumerate}[(a)]
\item It is easy to check that $<g_{i},g_{k+1}>_H=<g_{i},g_{k+2}>_H=0$ for $1\le i\le k$. 
\item Under the condition $(n-1)|k(q+1)$, we have that $(\alpha_1^{k(q+1)},\hdots, \alpha_n^{k(q+1)})=(1,\hdots,1,0)$, and so 
$<g_{k+1},g_{k+1}>_H=v_1^{q+1}+\hdots+ v_{n-1}^{q+1}=-v_{n}^{q+1}\not=0$.

\item It follows after taking the exponent of $\alpha_i$ modulo $n-1$ that $<g_{k+2},g_{k+2}>_H=\sum\limits_{l=1}^nv_l^{q+1}\alpha_l^{j}=0$, where $0< j\le n-2$. 
\item The fact that $<g_{k+1},g_{k+2}>_H=0$ follows from the same reasoning 
as in the proof of  point (d) of the case $(n-1)|(k+1)(q+1)$.
\end{enumerate}
\begin{itemize}
\item By choosing $\lambda_2\not=0$ and $\lambda_1\not=\pm v_n$, we obtain that 
$$
\begin{array}{ll}
G_{k+2}G_{k+2}^\dag&=L_{k+2}L_{k+2}^\dag+R_{k+2}R_{k+2}^\dag\\
&=\text{diag} (0,\hdots,0,\lambda_1^{q+1}-v_n^{q+1}, \lambda_2^{q+1}).
\end{array}
$$
It follows that $\dim(Hull_H(C))=k+2-\text{rank}(G_{k+2}G_{k+2}^\dag)=k$.
\item By choosing $\lambda_2\not=0$ and $\lambda_1=\pm v_n$, we obtain that $\dim(Hull_H(C))=k+2-\text{rank}(G_{k+2}G_{k+2}^\dag)=k+1$.
\end{itemize}
\end{proof}

By puncturing the last two coordinates of the codes in Theorem \ref{thm:embedding-new}, we obtain the following result.
\begin{cor}
Assume that $(n-1)|(q^2-1)$ and $1\le k\le \lfloor \frac{n+q-1}{q+1} \rfloor$. 
Then, there exists an MDS linear code $C$ with parameters
 $[n,k+2,\ge n-k-1]_{q^2}$ such that $\dim(Hull_H(C))=k+1$ if $(n-1)|(k+1)(q+1)$ or $(n-1)|k(q+1)$.
\end{cor}
Next, we consider the code with dimension $k+i$  strictly greater than $ \lfloor \frac{n-1+q}{q+1}\rfloor+2$, that is, the code $C_{k+i}$ with the following generator matrix $G_{k+i}$:
\begin{equation}
G_{k+i}=(L_{k+i}:R_{k+i}),
\label{eq:Gi}
\end{equation}
where
{\scriptsize
\begin{equation}
L_{k+i}=\left(
\begin{array}{cccccccc}
v_1&\hdots&v_{n}\\
v_1\alpha_1&\hdots&v_{n}\alpha_{n}\\
\vdots&\vdots&\vdots\\
v_1\alpha_1^{k-1}&\hdots&v_{n}\alpha_{n}^{k-1}\\
v_1\alpha_1^{k}&\hdots&v_{n}\alpha_{n}^{k}\\
v_1\alpha_1^{k+1}&\hdots&v_{n}\alpha_{n}^{k+1}\\
\vdots&\vdots&\vdots\\
v_1\alpha_1^{k+i-2}&\hdots&v_{n}\alpha_{n}^{k+i-2}\\
v_1\alpha_1^{k+i-1}&\hdots&v_{n}\alpha_{n}^{k+i-1}\\
\end{array}
\right),
R_{k+i}=\left(
\begin{array}{cc}
0&0\\
0&0\\
\vdots&\vdots\\
0&0\\
0&0\\
0&0\\
\vdots&\vdots\\
\lambda_1&0\\
\delta&\lambda_2\\
\end{array}
\right).
\end{equation}
}
In the sequel, we denote
\begin{equation}\label{eq:delta_t}
\delta_{t-1}:=<g_t,g_t>_H=\sum\limits_{j=1}^nv_j^{q+1}\alpha_j^{(t-1)(q+1)},
\end{equation}
where $g_t$ is the $t$-th row of the generator matrix of $GRS_{k+i}(\alpha,{\bf v})$,
\begin{equation}
\Delta_{k+i}=\{j|\delta_j\not =0,0\le j\le k+i-1\},
\end{equation}
and write $\sharp \Delta_{k+i}$ for the size of $\Delta_{k+i}$.

\begin{lem}
\label{lem:2} Let $q=p^m$ be a prime power and $n$ be a positive integer such that $(n-1)|(q^2-1)$. 
Let $\alpha_1,\hdots,\alpha_{n-1}$ be $n-1$ distinct elements of $\F_{q^2}$ such that $\alpha_i^{n-1}=1$. Let $k_0'$ be the smallest integer such that $(n-1)|k_0'(q+1)$. Assume that the parameters of the $GRS_{k_0'}(\alpha,{\bf v})$ code, where $\alpha=(\alpha_1,\hdots,\alpha_{n-1},\alpha_n)$  with $\alpha_n=0$ and ${\bf v}=(v_1,\hdots,v_n)$, satisfy $u_i=v_i^{q+1}$ for $1\le i\le n$.

Then, the code $GRS_{k_0'+1}(\alpha,{\bf v})$ code has a generator matrix $(g_1^\top,\hdots, g_{k'_0+1}^\top)$ satisfying the following properties:
\begin{enumerate}
\item $\langle g_i,g_j\rangle_H=0$ for any $1\le i\not=j\le k_0' +1$;
\item $\langle g_{k_0'+1},g_{k_0'+1}\rangle_H\not=0$ if $(n-1)|k_0'(q+1)$;
\item $\langle g_{k_0'+1},g_{k_0'+1}\rangle_H=0$ if $(n-1){\not|}k_0'(q+1)$.
\end{enumerate}
\end{lem}
\begin{proof} Assume that $u_i=v_i^{q+1}$ for $1\le i\le n$.
Set $z_i=v_i^{\frac{q+1}{2}}$ for $1\le i\le n$. Since $u_i={z_i}^2$ for $1\le i\le n$, for any $k'\le \frac{n}{2}$, the code $GRS_{k'}(\alpha,{\bf z})$, where ${\bf z}=(z_1,\hdots,z_n)$, is Euclidean self-orthogonal with parameters $[n,k']_{q^2}$. Consider a generator matrix of $GRS_{k'}(\alpha,{\bf z})$ with the following form:
\begin{equation}
{\cal G}'=\left(
\begin{array}{ccccc}
z_1&z_2&\hdots&z_{n}\\
z_1\alpha_1&z_2\alpha_2&\hdots&z_{n}\alpha_{n}\\
\vdots&\vdots&\vdots&\vdots\\
z_1\alpha_1^{k'-1}&z_2\alpha_2^{k'-1}&\hdots&v_{n}\alpha_{n}^{k'-1}\\
\end{array}
\right)
=\left(
\begin{array}{c}
{g}'_1\\
{g}'_2\\
\vdots\\
{g}'_{k'}\\
\end{array}
\right).
\end{equation}
By the Euclidean self-orthogonality of the code $GRS_{k'}(\alpha,{\bf z})$, it follows that $$\langle {g}'_1,{g}'_1\rangle_E=\sum\limits_{l=1}^n{z_l^2}=0,$$ and for $1\le i\le k',2\le j\le k'$,
$$\langle {g}'_i,{g}'_j\rangle_E=\sum\limits_{l=1}^n{z_l^2\alpha_l^{(i-1)+(j-1)}}=0,$$
equivalently,
\begin{equation}\label{eq:E0}
\sum\limits_{l=1}^nv_l^{q+1}=0,
\end{equation}
\begin{equation}\label{eq:E1}
\sum\limits_{l=1}^nv_l^{q+1}\alpha_l^{i}=0\text{ for }1\le i \le n-2.
\end{equation}
Denote the generator matrix of $GRS_{k_0'}(\alpha,{\bf v})$ by
\begin{equation}
{\cal G}=\left(
\begin{array}{ccccc}
v_1&v_2&\hdots&v_{n}\\
v_1\alpha_1&v_2\alpha_2&\hdots&v_{n}\alpha_{n}\\
\vdots&\vdots&\vdots&\vdots\\
v_1\alpha_1^{k_0'}&v_2\alpha_2^{k_0'}&\hdots&v_{n}\alpha_{n}^{k_0'}\\
\end{array}
\right)
=\left(
\begin{array}{c}
{g}_1\\
{g}_2\\
\vdots\\
{g}_{k_0'+1}\\
\end{array}
\right).
\end{equation}
Then, for $1\le i,j\le k_0'+1,$
\begin{equation}\label{eq:H-gigj}
\langle {g}_i,{g}_j\rangle_H=\sum\limits_{l=1}^n{v_l^{q+1}\alpha_l^{(i-1)+q(j-1)}}.
\end{equation}
Now, let us write
\begin{equation}\label{eq:B(t)}
(i-1)+q(j-1)=A\times (n-1)+B(i,j),
\end{equation}
with $A$ being a non-negative integer and $0\le B(i,j)<n-1$. 
We have that
$$
\begin{array}{ll}
\delta_{k_0'}=\langle {g}_{k_0'+1},{g}_{k_0'+1}\rangle_H&=\sum\limits_{l=1}^nv_l^{q+1}\alpha_l^{k_0'(q+1)}.\\
\end{array}
$$
It follows that $\delta_{k_0'}\not=0$ if $(n-1)|k_0'(q+1)$, and $\delta_{k_0'}=0$ if $(n-1){\not|}k_0'(q+1)$,
which proves point 2) and point 3), respectively. 

It follows from (\ref{eq:E1}) that, for $1\le i\not=j\le k_0'$,
$$
\begin{array}{ll}
\langle {g}_{i},{g}_{j}\rangle_H&=\sum\limits_{l=1}^nv_l^{q+1}\alpha_l^{i-1+q(j-1)}\\
&=\sum\limits_{l=1}^nv_l^{q+1}\alpha_l^{B(i,j)}\\
&=0.
\end{array}
$$
Since $k_0'$ is the smallest integer such that $(n-1)|k_0'(q+1)$, $(n-1){\not|}j(q+1)$ for $1\le j\le k_0'-1$, and we obtain that $\langle {g}_{j},{g}_{j}\rangle_H=0$ for $1\le j\le k_0-1$. This proves point 1).
\end{proof}

In the following discussion, we will explore the possible values to be taken by $k_0'$ (greater than that in Lemma \ref{lem:2}) such that the code $GRS_{k_0'+1}(\alpha,{\bf v})$ code has a generator matrix $(g_1^\top,\hdots, g_{k'_0+1}^\top)$ satisfying some desire properties.

Let $n$ be an odd integer and $k_0$ be the smallest integer such that $(n-1)|k_0(q+1)$. Then, from Lemma \ref{lem:2}, we have that
 \begin{equation}
 \begin{array}{l}
   \langle g_{k_0+1},g_{k_0+1}\rangle_H\not=0.\\
  \end{array}
 \end{equation}
Let $k_0'$ be an integer such that $k_0<k_0'$.
By writing 
\begin{equation}\label{eq:parity-Bkk}
(k_0'-1)+q(k_0'-1)=A\times (n-1)+B(k_0',k_0'),
\end{equation}
we obtain that $$(k_0'-1)+k_0'q=(k_0'-1)+q(k_0'-1)=A\times (n-1)+B(k_0',k_0')+q.$$ 
Note that for both $k_0'$ odd and even, $B(k_0',k_0')$ is even since $n-1$ is even.
By the same reasoning as that in the proof of Lemma \ref{lem:2}, it follows that $\langle g_{k_0'},g_{k_0'+1}\rangle_H=0$, and thus $1\le B(k_0',k_0')+q\le n-2$ with $B(k_0',k_0')$ even.
Thus, one can choose $k_0'$ in the following range:
\begin{equation}\label{eq:k01}
  2(k_0'-1)+q\le n-2,
 \end{equation} 
 Since $\langle g_{k_0+1},g_{k_0+1}\rangle_H=\langle g_{(k_0+1)+(q-1)},g_{(k_0+1)+(q-1)}\rangle_H\not=0$, we obtain that 
$\langle g_{k_0+q-1},g_{k_0+q-1}\rangle_H=0$. Hence, one
can choose an upper bound for $k_0'$ as follows:
\begin{equation}\label{eq:k02}
k_0'\le k_0+q-1.
\end{equation}
Therefore, from (\ref{eq:k01}) and (\ref{eq:k02}), one can choose  $k_0'= \min \left(\frac{n-q}{2}, k_0+q-1\right)$ so that the code $GRS_{k_0'+1}(\alpha,{\bf v})$ has a generator matrix $(g_1^\top,\hdots, g_{k'_0+1}^\top)$ satisfying the following properties:
\begin{enumerate}
\item $\langle g_i,g_j\rangle_H=0$ for any $1\le i\not=j\le k_0' +1$;
\item $\langle g_{k_0'+1},g_{k_0'+1}\rangle_H\not=0$ if $(n-1)|k_0'(q+1)$;
\item $\langle g_{k_0'+1},g_{k_0'+1}\rangle_H=0$ if $(n-1){\not|}k_0'(q+1)$.
\end{enumerate}

Combining the above discussion with Lemma \ref{lem:2}, we get the following result.
\begin{thm}\label{thm:extended} Assume that $n$ is odd and $(n-1)|(q^2-1)$.
Put $k= \lfloor \frac{n+q-1}{q+1} \rfloor$, $i_{\min}=\min (\lfloor \frac{n-q}{2}\rfloor-k,q-1)$ and $K=k+{i_{\min}}+1$. Assume further that $(n-1)|k(q+1)$.
\begin{enumerate}
\item If $(n-1){\not|}(K-1)(q+1)$ and $(n-1){\not|}(K-2)(q+1)$, then there exists a linear code $C$ with parameters
 $[n+2,K, n+2-K]_{q^2}$ with $\dim(Hull_H(C))=K-\sharp \Delta_{K}-2$;
\item 
If $(n-1){\not|}(K-1)(q+1)$ and $(n-1)|(K-2)(q+1)$, then there exists a linear code $C$ with parameters
 $[n+2,K, n+2-K]_{q^2}$ with $\dim(Hull_H(C))=K-\sharp \Delta_{K}-1$;  
 \item 
If $(n-1)|(K-1)(q+1)$ and $(n-1){\not|}(K-2)(q+1)$, then there exists a linear code $C$ with parameters
 $[n+2,K, n+2-K]_{q^2}$ with $\dim(Hull_H(C))=K-\sharp \Delta_{K}-1$; 
 \item If $(n-1)|(K-1)(q+1)$ and $(n-1)|(K-2)(q+1)$, then there exists a linear code $C$ with parameters
 $[n+2,K, n+2-K]_{q^2}$ with $\dim(Hull_H(C))=K-\sharp \Delta_{K}$.
 \end{enumerate}
\end{thm}
\begin{proof} 
Take $U_{n-1}=\{ \alpha\in \F_{q^2}|\alpha^{n-1}=1\}$, and label the elements of $U_{n-1}$ as $\alpha_1,\hdots,\alpha_{n-1}$, as well as set $\alpha_n=0$. Put $\alpha=(\alpha_1,\hdots, \alpha_n)$ and ${\bf v}=(v_1,\hdots,v_n)$. Since $(n-1)|(q^2-1)$, $1\le k\le \lfloor \frac{n+q-1}{q+1} \rfloor$ and  $u_i=v_i^{q+1}$ for $1\le i\le n$, the code code $GRS_{k}(\alpha,{\bf v})$ has a generator matrix satisfying properties 1)-3) of Lemma \ref{lem:2}.
Consider the code $C_{K}$ with generator matrix $G_{K}$ ($K=k+{i_{\min}}+1$) written as in (\ref{eq:Gi}) such that $\delta=0$. 
It follows that the code $C_{K}$ also has a generator matrix satisfying the properties 1)-3) of Lemma \ref{lem:2}.\\
\begin{enumerate}
\item If $(n-1){\not|}(K-1)(q+1)$ and $(n-1){\not|}(K-2)(q+1)$, then we can take $\lambda_1$ and $\lambda_2$ to be any elements of $\F_{q^2}^*$. Hence, $\text{rank}(G_{K}G_{K}^\dag)=\sharp \Delta_{K}+2$;
\item If $(n-1){\not|}(K-1)(q+1)$ and $(n-1)|(K-2)(q+1)$, then we can take $\lambda_1$ to be any element of $\F_{q^2}^*$ and $\lambda_2\not=0,\pm v_n$. Hence, $\text{rank}(G_{K}G_{K}^\dag)=\sharp \Delta_{K}+1$; 
 \item If $(n-1)|(K-1)(q+1)$ and $(n-1){\not|}(K-2)(q+1)$, then we can take $\lambda_2$ to be any element of $\F_{q^2}^*$ and $\lambda_1\not=0,\pm v_n$. Hence, $\text{rank}(G_{K}G_{K}^\dag)=\sharp \Delta_{K}+1$;
 \item If $(n-1)|(K-1)(q+1)$ and If $(n-1)|(K-2)(q+1)$,  Then, we can take $\lambda_1\not=0,\pm v_n,$ and $\lambda_2\not=0,\pm v_n$. Hence, $\text{rank}(G_{K}G_{K}^\dag)=\sharp \Delta_{K}$.
 \end{enumerate}
\end{proof}
By puncturing the last two coordinates of the code in Theorem \ref{thm:extended}, we obtain the following result.
\begin{cor}\label{cor:extended} Assume that $n$ is odd and $(n-1)|(q^2-1)$. Put $k= \lfloor \frac{n+q-1}{q+1} \rfloor$ and $i_{\min}=\min (\lfloor \frac{n-q}{2}\rfloor-k,q-1)$. Assume further that $(n-1)|k(q+1)$. Then, for any $1\le k'\le k+{i_{\min}}+1$, there exists an MDS linear code $C$ with parameters
 $[n,k']_{q^2}$  such that $\dim(Hull_H(C))=k'-\sharp \Delta_{k'}$.
\end{cor}

\begin{exam}\label{ex:2} Take $q=7$, $n=17$ and $\theta$ as a primitive element of $\F_{7^2}$. Then, from \cite{SokQSC}, there exists a Hermitian self-orthogonal code $C$ with parameters $[17,2,16]_{49}$.
If $\lambda_1\not=\theta^3$ and $\lambda_2\not=0$, then by Theorem \ref{thm:extended}, the code $C_6$ has parameters $[19,6,13]_{49}$ with $\dim Hull_H(C_6)=3$. Its generator matrix $G_6=(A_6|B_6)$ is given as follows:
{\scriptsize
$
A_6=\left(
\begin{array}{llllllllllllllllllllllll}
\theta^{45}&\theta^{46}&\theta^{46}&\theta^{46}&\theta^{46}&\theta^{46}&\theta^{46}&\theta^{46}&\theta^{46}&\theta^{46}\\
0&\theta^{22}&\theta^{25}&\theta^{28}&\theta^{31}&\theta^{34}&\theta^{37}&5&\theta^{43}&\theta^{46}\\
0&\theta^{46}&\theta^4&\theta^{10}&2&\theta^{22}&\theta^{28}&\theta^{34}&5&\theta^{46}\\
0&\theta^{22}&\theta^{31}&5&\theta&\theta^{10}&\theta^{19}&\theta^{28}&\theta^{37}&\theta^{46}\\
0&\theta^{46}&\theta^{10}&\theta^{22}&\theta^{34}&\theta^{46}&\theta^{10}&\theta^{22}&\theta^{34}&\theta^{46}\\
0&\theta^{22}&\theta^{37}&\theta^4&\theta^{19}&\theta^{34}&\theta&2&\theta^{31}&\theta^{46}\\
\end{array}
\right),
$
$B_6=\left(
\begin{array}{llllllllllllllllllllllll}
\theta^{46}&\theta^{46}&\theta^{46}&\theta^{46}&\theta^{46}&\theta^{46}&\theta^{46}&0&0\\
\theta^{46}&\theta^4&\theta^7&\theta^{10}&\theta^{13}&2&\theta^{19}&0&0\\
\theta^4&\theta^{10}&2&\theta^{22}&\theta^{28}&\theta^{34}&5&\lambda_1&0\\
\theta^7&2&\theta^{25}&\theta^{34}&\theta^{43}&\theta^4&\theta^{13}&0&\lambda_2\\
\theta^{10}&\theta^{22}&\theta^{34}&\theta^{46}&\theta^{10}&\theta^{22}&\theta^{34}&0&0\\
\theta^{13}&\theta^{28}&\theta^{43}&\theta^{10}&\theta^{25}&5&\theta^7&0&0\\
\end{array}
\right).
$
}
It is easy to check that $\delta_2=\langle g_3,g_3\rangle\not=0,\delta_3=\langle g_4,g_4\rangle=0,\delta_4=\langle g_5,g_5\rangle\not=0,\delta_5=\langle g_6,g_6\rangle=0$.
By puncturing the last two coordinates of $G_6$, one obtains an MDS linear code $C'_6$ with parameters $[17,6,12]_{49}$ such that $\dim Hull_H(C'_6)=4$. One can obtain an MDS linear code $C'_5$ with parameters $[17,5,13]_{49}$ such that $\dim Hull_H(C'_5)=3$ by puncturing the last two coordinates and taking the first five rows of $G_6$ as a generator matrix of $C'_5$.
We summarize the parameters of the codes and their Hermitian hull dimensions obtained from this example as follows:
$$
\begin{array}{c|c|c}
\text{Parameters} &\text{Reference}&\ell\\
\hline
[19,6,13]&\text{Theorem \ref{thm:extended}}&4\\

[19,6,13]&\text{Theorem \ref{thm:extended}}&3\\

[19,5,14]&\text{Theorem \ref{thm:extended}}&3\\

[19,5,14]&\text{Theorem \ref{thm:extended}}&2\\

[19,4,15]&\text{Theorem \ref{thm:extended}}&3\\

[19,4,15]&\text{Theorem \ref{thm:extended}}&2\\

[19,3,16]&\text{Theorem \ref{thm:extended}}&2\\

[17,6,12]&\text{Corollary \ref{cor:extended}}&4\\

[17,5,13]&\text{Corollary \ref{cor:extended}}&3\\

[17,4,14]&\text{Corollary \ref{cor:extended}}&3\\

[17,3,15]&\text{Corollary \ref{cor:extended}}&2\\

\end{array}
$$
\end{exam}

By puncturing the zero coordinate among the last two coordinates of the code in Theorem \ref{thm:extended}, we obtain the following result.
\begin{cor}\label{cor:extended-1} Assume that $n$ is odd and $(n-1)|(q^2-1)$. Put $k= \lfloor \frac{n+q-1}{q+1} \rfloor$ and $i_{\min}=\min (\lfloor \frac{n-q}{2}\rfloor-k,q-1)$. Assume further that $(n-1)|k(q+1)$. Then, for any $1\le k'\le k+{i_{\min}}+1$, there exists an MDS linear code $C$ with parameters
 $[n+1,k', n-k'+2]_{q^2}$  such that $\dim(Hull_H(C))=k'-\sharp \Delta_{k'}$.
\end{cor}
We have seen that the codes constructed in Theorem \ref{thm:extended}, Corollary \ref{cor:extended}, and Corollary \ref{cor:extended-1} have restrictive Hermitian hull dimensions. However, by slightly modifying the above constructions, 
can obtain codes with more flexible hull dimensions. 
\begin{thm}\label{thm:extended-2} Assume that $n$ is odd and $(n-1)|(q^2-1)$. Put $k= \lfloor \frac{n+q-1}{q+1} \rfloor$ and $i_{\min}=\min (\lfloor \frac{n-q}{2}\rfloor-k,q-1)$. Assume further that $(n-1)|k(q+1)$. Then, for any $1\le k'\le k+{i_{\min}}+1$, there exists an MDS linear code $C$ with parameters
 $[n,k',n-k'+1]_{q^2}$ such that $\dim(Hull_H(C))=j$ for $0\le j\le k'-\sharp \Delta_{k'}$.
 \end{thm}
 
 \begin{proof} Take $a \in \mathbb{F}_{q}^{*}$ with $a^{q+1} \neq 1$.
Let $1\le k'\le k+{i_{\min}}+1$ and $0\le r \le k'-1$. 
From Corollary \ref{cor:extended}, there exists an MDS linear code $C'$ with parameters $[n,k']$ such that $\dim Hull_H(C')=k'-\sharp \Delta_{k'}$. Let $G'$ be a generator matrix of $C'$. 
From Lemma  \ref{lem:2}, $G'$ can be, up to equivalence, written as:
$$
G'=\left(
\begin{array}{llllllllll}
1&0&\hdots&0&0&0&A'_1\\
0&1&\hdots&0&0&0&A'_2\\
\vdots&\vdots&&\vdots&\vdots&\vdots&\vdots\\
0&0&\hdots&1&0&0&A'_{r}\\
\vdots&\vdots&&\vdots&\vdots&\vdots&\vdots\\
0&0&\hdots&0&0&1&A'_{k'}\\
\end{array}
\right)=
\left(
\begin{array}{llllllllll}
g'_1\\
g'_2\\
\vdots\\
g'_r\\
\vdots\\
g'_{k'}\\
\end{array}
\right)
$$
such that $\langle g'_i,g'_j \rangle_H=0$ for $1\le i\not=j\le k'$, $\langle g'_i,g'_i \rangle_H=0$ for $1\le i\le k'-\sharp \Delta_{k'}$ and $\langle g'_i,g'_i \rangle_H\not=0$ for $1\le i\le k'-\sharp \Delta_{k'}+1,\hdots, k'$. 
Take a code $C''$ with the following generator matrix:
$$
G''=\left(
\begin{array}{llllllllll}
a&0&\hdots&0&0&0&0&A'_1\\
0&a&\hdots&0&0&0&0&A'_2\\
\vdots&\vdots&&\vdots&\vdots&\vdots&\vdots&\vdots\\
0&0&\hdots&a&0&0&0&A'_{r}\\
0&0&\hdots&0&1&0&0&A'_{r+1}\\
\vdots&\vdots&&\vdots&\vdots&\vdots&\vdots&\vdots\\
0&0&\hdots&0&0&0&1&A'_{k'}\\
\end{array}
\right)
$$
It is easy to see that $\text{rank}(G''G''^\dag)=r+\sharp \Delta_{k'}$ and thus the result follows from Lemma \ref{lem:hull-H}.
Hence the code $C''$ has parameters $[n,k',n-k'+1]_{q^2}$ such that $\dim Hull_H(C'')=j$ for $j=0,\hdots, k'-\sharp \Delta_{k'}$.\end{proof}
\begin{cor}\label{cor:extended-3} Assume that $n$ is odd and $(n-1)|(q^2-1)$. Put $k= \lfloor \frac{n+q-1}{q+1} \rfloor$ and $i_{\min}=\min (\lfloor \frac{n-q}{2}\rfloor-k,q-1)$. Assume further that $(n-1)|k(q+1)$. Then, for any $1\le k'\le k+{i_{\min}}+1$, there exists an MDS linear code $C$ with parameters
 $[n+1,k',n+2-k']_{q^2}$  such that $\dim(Hull_H(C))=j$ for $0\le j\le k'-\sharp \Delta_{k'}$.
\end{cor}
\begin{proof} The proof follows with the same reasoning as that in Theorem \ref{thm:extended-2} by taking $C'$ to be the code in Corollary \ref{cor:extended-1}.
\end{proof}

\begin{cor}\label{cor:extended-4} Assume that $n$ is odd and $(n-1)|(q^2-1)$. Put $k= \lfloor \frac{n+q-1}{q+1} \rfloor$ and $i_{\min}=\min (\lfloor \frac{n-q}{2}\rfloor-k,q-1)$. Assume further that $(n-1)|k(q+1)$. Then, for any $1\le k'\le k+{i_{\min}}+1$, there exists an AMDS linear code $C$ with parameters
 $[n+2,k', n-k'+2]_{q^2}$  such that $\dim(Hull_H(C))=j$ for $0\le j\le k'-\sharp \Delta_{k'}$.
\end{cor}

\begin{proof} The proof follows with the same reasoning as that in Theorem \ref{thm:extended-2} by taking $C'$ to be the code in Theorem \ref{thm:extended}.
\end{proof}

\begin{thm}\label{thm:extended-3} Let $q=p^m$ be a prime power and $N|(q^2-1)$. Put $n=(t+1)N+1$, $k= \lfloor \frac{n+q-1}{q+1} \rfloor$ and $i_{\min}=\min (\lfloor \frac{n-q}{2}\rfloor-k,q-1)$, where $1\le t\le \frac{q-1}{n_2}-2$ and $n_2=\frac{N}{\gcd (N,q+1)}$. Assume that $n$ is odd and $(n-1)|k(q+1)$ or $(n-1)|(k+1)(q+1)$.
Then, for any $1\le k'\le k+{i_{\min}}+1$, there exists an MDS linear code $C$ with parameters
 $[n,k']_{q^2}$ such that $\dim(Hull_H(C))=j$ for $0\le j\le k'-\sharp \Delta_{k'}$.
 \end{thm}

\begin{proof}Let $U$ be defined as in Eq. (\ref{eq:multi-coset}). For $k<k'\le k+{i_{\min}}+1$, consider the code $GRS_{k'}(\alpha,{\bf v})$ with $\alpha=(0,{\bf a},\alpha_1{\bf a},\hdots,\alpha_t{\bf a})$ and ${\bf v}=(v_0,{\bf v}^{(0)},{\bf v}^{(1)},\hdots,{\bf v}^{(t)})$ being vectors of length $n=(t+1)N+1$, where ${\bf a}=(a_1,\hdots,a_N)$ and $v_i^{q+1}=u_i$ for $1\le i\le n$. Denote the generator matrix $GRS_{k'}(\alpha,{\bf v})$ by ${\cal G}_{k'}$, where ${\cal G}_{k'}$ is given as follows:
{\tiny
$$
\left(
\begin{array}{cc|c|c|c}
v_0&{\bf v}^{(0)}&{\bf v}^{(1)}&\hdots&{\bf v}^{(t)}\\
0&{\bf v}^{(0)}* {\bf a}&{\bf v}^{(1)}*(\alpha_1{\bf a})&\hdots&{\bf v}^{(t)}*(\alpha_t{\bf a})\\
\vdots&\vdots&\vdots&\vdots&\vdots\\
0&\left({\bf v}^{(0)}* {\bf a}\right)^{k'-1}&\left({\bf v}^{(1)}*(\alpha_1{\bf a})\right)^{k'-1}&\hdots&\left({\bf v}^{(t)}*(\alpha_t{\bf a})\right)^{k'-1}\\
\end{array}
\right),
$$
}
where $(x_1,\hdots,x_N)*(y_1,\hdots,y_N)$ denotes $(x_1y_1,\hdots,x_Ny_N)$.
Set $z_i=v_i^{\frac{q+1}{2}}$ for $1\le i\le n$. Since $u_i={z_i}^2$ for $1\le i\le n$, for any $K'\le \frac{n}{2}$, the code $GRS_{K'}(\alpha,{\bf z})$, where $\alpha=(a_1,\hdots,a_n)$ and ${\bf z}=(z_1,\hdots,z_n)$, is Euclidean self-orthogonal with parameters $[n,K']_{q^2}$. 
Consider a generator matrix of $GRS_{K'}(\alpha,{\bf z})$ with the following form:
\begin{equation}
{\cal G}'_{K'}=\left(
\begin{array}{ccccc}
z_1&z_2&\hdots&z_{n}\\
z_1a_1&z_2a_2&\hdots&z_{n}a_{n}\\
\vdots&\vdots&\vdots&\vdots\\
z_1a_1^{K'-1}&z_2a_2^{K'-1}&\hdots&z_{n}a_{n}^{K'-1}\\
\end{array}
\right)
=\left(
\begin{array}{c}
{g}'_1\\
{g}'_2\\
\vdots\\
{g}'_{K'}\\
\end{array}
\right).
\end{equation}
By the Euclidean self-orthogonality of the code, it follows that for $1\le i\le K'$ and $2\le j\le K'$,
\begin{equation*}\label{eq:E-gigj-2}
\begin{array}{ll}
\langle {g}'_i,{g}'_j\rangle_E&=\sum\limits_{l=1}^n{z_l^2a_l^{(i-1)+(j-1)}}\\
&=\sum\limits_{l=1}^Nz_l^2a_l^{(i-1)+(j-1)}+\sum\limits_{l=1}^Nz_{N+l}^2(\alpha_1a_l)^{(i-1)+(j-1)}\\
&~~+\cdots +\sum\limits_{l=1}^Nz_{Nt+l}^2(\alpha_ta_l)^{(i-1)+(j-1)}\\
&=0.\\
\end{array}
\end{equation*}
It follows that for $1\le i\le (t+1)N-2$,
\begin{equation}\label{eq:E-multi}
\sum\limits_{l=1}^Nz_l^2a_l^{i}+\sum\limits_{l=1}^Nz_{N+l}^2(\alpha_1a_l)^{i}
+\cdots +\sum\limits_{l=1}^Nz_{Nt+l}^2(\alpha_ta_l)^{i}=0.
\end{equation}

From the generator matrix ${\cal G}_{k'}$, we have that, for $1\le i\le k'$ and $2\le j\le k',$
\begin{equation}\label{eq:H-multi}
\begin{array}{ll}
\langle {g}_i,{g}_j\rangle_H&=\sum\limits_{l=1}^n{v_l^{q+1}a_l^{(i-1)+q(j-1)}}\\
&=\sum\limits_{l=1}^Nz_l^2a_l^{(i-1)+q(j-1)}+\sum\limits_{l=1}^Nz_{N+l}^2(\alpha_1a_l)^{(i-1)+q(j-1)}\\
&+\cdots +\sum\limits_{l=1}^Nz_{tN+l}^2(\alpha_ta_l)^{(i-1)+q(j-1)}.
\end{array}
\end{equation}
By writing the exponent $(i-1)+q(j-1)$ in the summand (\ref{eq:H-multi}) as $(i-1)+q(j-1)=A\times (t+1)N+B(i,j)$ with $0\le B(i,j)< (t+1)N$, we obtain that
\begin{equation}\label{eq:H-multi-2}
\begin{array}{ll}
\langle {g}_i,{g}_j\rangle_H&=\sum\limits_{l=1}^n{v_l^{q+1}a_l^{B(i,j)}}\\
&=\sum\limits_{l=1}^Nz_l^2a_l^{B(i,j)}+\sum\limits_{l=1}^Nz_{N+l}^2(\alpha_1a_l)^{B(i,j)}\\
&~~+\cdots +\sum\limits_{l=1}^Nz_{tN+l}^2(\alpha_ta_l)^{B(i,j)}.
\end{array}
\end{equation}
It follows from (\ref{eq:E-multi}) that for $1\le i\le k'$ and $2\le j\le k',$
$\langle {g}_i,{g}_j\rangle_H=0$. 
The rest follows with the same reasoning as that in the proof of Theorem \ref{thm:extended-2}.
\end{proof}

\begin{exam} Take $q=5$, $\theta$ a primitive element of $\F_{5^2}$, $N=6$ and $t=2$, we can construct an MDS $[19,8,12]_{5^2}$ code $C_8$ with $\dim Hull_H(C_8)=3$. Its generator matrix is $G_8=(A_8|B_8)$, where  
{\scriptsize
$$
A_8=\left(
\begin{array}{llllllllll}
\theta^{23}&\theta^{21}&1&\theta^{22}&\theta^{21}&1&\theta^{22}&\theta^{21}&1&\theta^{22}\\
0&\theta^{9}&\theta^{13}&4&\theta^{13}&\theta^{17}&\theta^{16}&\theta^{17}&\theta^{21}&\theta^{20}\\
0&\theta^{21}&\theta^{2}&\theta^{2}&\theta^{5}&\theta^{10}&\theta^{10}&\theta^{13}&3&3\\
0&\theta^{9}&\theta^{15}&\theta^{16}&\theta^{21}&\theta^{3}&\theta^{4}&\theta^{9}&\theta^{15}&\theta^{16}\\
0&\theta^{21}&\theta^{4}&2&\theta^{13}&\theta^{20}&\theta^{22}&\theta^{5}&4&\theta^{14}\\
0&\theta^{9}&\theta^{17}&\theta^{20}&\theta^{5}&\theta^{13}&\theta^{16}&\theta&\theta^{9}&4\\
0&\theta^{21}&2&\theta^{10}&\theta^{21}&2&\theta^{10}&\theta^{21}&2&\theta^{10}\\
0&\theta^{9}&\theta^{19}&1&\theta^{13}&\theta^{23}&\theta^{4}&\theta^{17}&\theta^{3}&\theta^{8}\\
\end{array}
\right),
$$
$$
B_8=\left(
\begin{array}{llllllllll}
\theta^{21}&1&\theta^{22}&\theta^{21}&1&\theta^{22}&\theta^{21}&1&\theta^{22}\\
\theta^{21}&\theta&1&\theta&\theta^{5}&\theta^{4}&\theta^{5}&\theta^{9}&\theta^{8}\\
\theta^{21}&\theta^{2}&\theta^{2}&\theta^{5}&\theta^{10}&\theta^{10}&\theta^{13}&3&3\\
\theta^{21}&\theta^{3}&\theta^{4}&\theta^{9}&\theta^{15}&\theta^{16}&\theta^{21}&\theta^{3}&\theta^{4}\\
\theta^{21}&\theta^{4}&2&\theta^{13}&\theta^{20}&\theta^{22}&\theta^{5}&4&\theta^{14}\\
\theta^{21}&\theta^{5}&\theta^{8}&\theta^{17}&\theta&\theta^{4}&\theta^{13}&\theta^{21}&1\\
\theta^{21}&2&\theta^{10}&\theta^{21}&2&\theta^{10}&\theta^{21}&2&\theta^{10}\\
\theta^{21}&\theta^{7}&4&\theta&\theta^{11}&\theta^{16}&\theta^{5}&\theta^{15}&\theta^{20}\\
\end{array}
\right).
$$
}
It can be easily checked that $G_8G_8^\dag=\textnormal{diag} (0,0,0,1,2,0,0,1)$, and thus $\langle g_i,g_j\rangle_H=0$ for any $1\le i\not=j\le 8$ and $\langle g_4,g_4\rangle_H\not=0$, $\langle g_5,g_5\rangle_H\not=0$, $\langle g_8,g_8\rangle_H\not=0$. Hence, for any $3\le k'\le 8$, there exists an MDS $[19,k']_{5^2}$ code $C'$ with $\dim Hull_H(C')=j$ for $0\le j\le k'-3$.
\end{exam}
\section{Application to EAQECCs}\label{section:application}
In this section, we construct EAQECCs from classical linear codes using the method proposed by Wilde {\em et al. }\cite{WilBru}. The next subsection introduces some basic notions about quantum codes, especially the notions of the entanglement-assisted quantum error-correcting codes.
\subsection{Quantum codes}
Let  $({\mathbb C})^{\bigotimes n}$ ($\cong \mathbb{C}^{q^n}$) be the Hilbert space over the complex field $\mathbb{C}$ of dimension $q^n$. A $q$-ary quantum code $Q$ of length $n$ is a subspace of $\mathbb{C}^{q^{n}}$ with dimension $K\ge 1.$
Let $\{|\textbf{a}\rangle =|a_{1}\rangle\bigotimes|a_{2}\rangle\bigotimes\hdots\bigotimes|a_{n}\rangle: (a_{1}, a_{2}, \ldots, a_{n}) \in \mathbb{F}^{n}_{q}\}$ be a basis of $\mathbb{C}^{q^{n}}$. The inner product of two quantum states
$
|\phi_1\rangle=\sum\limits_{{\bf a}\in \F_{q}^{n}}\phi_1({\bf a})|{\bf a}\rangle 
\text { and } |\phi_2\rangle=\sum\limits_{{\bf a}\in \F_{q}^{n}}\phi_2({\bf a})|{\bf a}\rangle
$
is defined by
\begin{equation*}
\langle \phi_1|\phi_2\rangle=\sum\limits_{{\bf a}\in \F_{q}^{n}}
\overline{\phi_1({\bf a})}\phi_2({\bf a})\in {\mathbb C} \text{,} 
\end{equation*}
where  $\overline{\phi_1({\bf a})}$ is the complex conjugate of 
$\phi_1({\bf a})$. Two quantum states $|\phi_1\rangle$ and $|\phi_2\rangle$ are said to be
\textit{orthogonal} if $\langle \phi_1|\phi_2\rangle=0$.
Let $\zeta_{p}$ be a complex primitive $p$-th root of unity, and denote the trace function from $\mathbb{F}_{q}$ to $\mathbb{F}_{p}$ by $tr(.)$.

The rules of $X(\textbf{a})$ and $Z(\textbf{b})$ on $|\textbf{v}\rangle \in \mathbb{C}^{q^{n}}$ ($\textbf{v} \in \mathbb{F}^{n}_{q}$) are given as
\begin{equation}
X(\textbf{a})|\textbf{v}\rangle=|\textbf{v}+\textbf{a}\rangle
\textnormal{ and }
Z(\textbf{b})|\textbf{v}\rangle=\zeta_{p}^{tr(\langle \textbf{v}, \textbf{b} \rangle_{E})}|\textbf{v}\rangle,
\label{eq:action}
\end{equation}
respectively.
For two vectors ${\bf a}=(a_1,\ldots,a_{n}),{\bf b}=(b_1,\ldots,b_{n}) \in \F_{q}^{n}$, we can write, from (\ref{eq:action}), $X({\bf a}) = X(a_{1}) \otimes \ldots \otimes X(a_{n})$ and
$Z({\bf b}) = Z(b_{1}) \otimes \ldots \otimes Z(b_{n})$ for the tensor product of
$n$ (error) operators. The set ${\cal E}_{n}=\{ X({\bf a})Z({\bf b}) : {\bf a},{\bf b} \in \F_{q}^{n} \}$
is an error basis on $\mathbb{C}^{q^{n}}$.
The error group $G_{n}$ 
is defined by
\begin{equation*}
 G_{n}:=\{\zeta_{p}^{t} X({\bf a}) Z({\bf b}) : {\bf a},{\bf b} \in \F_{q}^{n}, t \in \F_{p} \} \text{.}
\end{equation*}
For $E = \zeta_{p}^{t} X({\bf a}) Z({\bf b}) \in G_{n}$, the \textit{quantum weight}
${\bf wt}_{Q}(E)$ of $E$ is the number of coordinates such that 
$(a_{i},b_{i}) \neq (0,0)$. A quantum code $Q$ with dimension $K\ge 2$ is said to  detect $d-1$ quantum errors ($d\geq1$) if, for any pair $|\phi_1\rangle$ and $|\phi_2\rangle$ in
$Q$ with $\langle \phi_1 |\phi_2 \rangle=0$ and any $E \in G_{n}$ with ${\bf wt}_{Q}(E) \leq d-1$, 
$|\phi_1\rangle$ and $E|\phi_2\rangle$ are orthogonal. In this case, such a code is called symmetric (\cite{AK01,Ket Kla}) and its parameters is written as $((n,K,d))_q$ or $[[n,k,d]]_q$ where $k=\log_qK$.

For $S$ being an abelian subgroup of $G_n$, the quantum stabilizer codes $C(S)$ is defined as
$$
C(S):=\{|\phi \rangle : E|\phi \rangle=|\phi \rangle,\forall E\in S\}.
$$
The authors \cite{Cal Rai,AK01,Ket Kla} showed that such codes can be constructed from classical linear codes with some properties of self-orthogonality.
However, the above construction method does not work anymore if the subgroup $S$ of $G_n$ is non-abelian. Brun \emph{et al.} \cite{BDH06} have improved the construction by introducing the so-called entanglement-assisted quantum error-correcting codes (EAQECCs). In their method, $S$ is extended to be a new abelian subgroup in a larger error group, and furthermore they assumed that both sender and receiver shared a specific amount of pre-existing entangled bits, which was not subject to errors. 

We use $[[n, k, d; c]]_{q}$ to denote a $q$-ary $[[n, k, d]]_{q}$ quantum code that utilizes $c$ pre-shared entanglement pairs. 
For $c=0$, an $[[n, k, d; 0]]_{q}$ EAQECC is equivalent to a quantum stabilizer code \cite{AK01}.

 \subsection{Constructions of new EAQECCs}
 In this subsection, we provide constructions of new EAQECCs, especially, the EAQECCs whose parameters satisfy some constraints.
 The Singleton bound for an EAQECC with parameters $[[n, k, d; c]]_{2}$ was given by Lai {\em et al. }\cite{LA18}. 
 Later, Grassl \cite{Grassl} found some counter examples of the bound given by \cite{LA18} when the minimum distance of the code is strictly greater than $\frac{n+2}{2}$, and the new bounds have been determined by Grassl {\em et al. }\cite{GHW}.

\begin{prop}[\cite{GHW}]\label{lem4.1}
  For any $[[n, k, d; c]]_{q}$-EAQECC, we have  
\begin{equation}\label{eq:bound1}
  k\le c+ \max \{0, n- 2d+2\},
\end{equation}
\begin{equation}\label{eq:bound0}
  k\le n-d+1,
\end{equation}
\begin{equation}\label{eq:bound2}
   k\le \frac{n-d+1}{3d-3-n}(c+2d-2-n)\text{ if } d \ge \frac{n+2}{2}.
\end{equation}
\end{prop}

When one of the bounds (\ref{eq:bound1})-(\ref{eq:bound2}) meets with equality, the EAQECC is called MDS.

\begin{prop} \textnormal{(\cite{BDH06})}\label{lem:Q-construction}
Let $P$ be the parity check matrix of an $[n, k, d]_{q^2}$ code $C$.  Then, there exists an $[[n, 2k - n + c, d; c]]_{q}$ EAQECC $\cal Q$, where $c =\text{rank}(PP^{\dag})$ is the required number of maximally entangled states. In particular, if $C$ is an MDS code and $d \leq \frac{n+2}{2}$, then $\mathcal{Q}$ is an MDS EAQECC.
\end{prop}

By combining Proposition \ref{lem:Q-construction} and Lemma \ref{lem:hull-H} together, $q$-ary EAQECCs can be constructed from a classical $q^2$-linear code as follows.
\begin{lem}\textnormal{(\cite{GJG18})}\label{lem:GJG18-construction}
Let $C$ be a linear code with parameters $[n, k, d]_{q^2}$ and $C^{\perp_H}$ its Hermitian dual  with parameters $[n,k,d']_{q^2}$. Assume that $\dim (Hull_H(C))=\ell$. Then, there exist an $[[n, k-\ell, d;n-k-\ell]]_{q}$ EAQECC and an $[[n, n-k-\ell, d';k-\ell]]_{q}$ EAQECC.
\end{lem}

It is easy to check that if a code $C$ is a linear code with Hermitian hull dimension $\ell$, then so is its dual $C^{\perp_H}$.
By applying Lemma \ref{lem:GJG18-construction} to the codes constructed in Theorem \ref{thm:TGRS-hermitian}, we obtain the following result.

\begin{thm}\label{thm:Q0}
Let $q=p^m$ be a prime power and $1\le k\le \lfloor \frac{n}{q+1} \rfloor$. Assume that one of the following conditions holds:
\begin{enumerate}
\item $(n-1)|(q^2-1)$;
\item $n=tq$, $1\le t\le q-1$;
\item $n=(t+1)N+1$, $N|(q^2-1)$, $n_2=\frac{N}{\gcd (N,q+1)}$, $1\le t\le \frac{q-1}{n_2}-2$.
\end{enumerate}
Then, for any $0\le \ell \le k-1$, there exist an $[[n, k-\ell, \ge n-k;n-k-\ell]]_{q}$ EAQECC and an $[[n, n-k-\ell, \ge k;k-\ell]]_{q}$ EAQECC.
\end{thm}

\begin{rem} EAQECCs obtained from Theorem \ref{thm:Q0} may have  parameters  overlapped with those in \cite{FangFuLiZhu} if the classical codes $[n,k]_{q^2}$, used to construct the EAQECCs, are MDS. 
\end{rem}

By applying Lemma \ref{lem:GJG18-construction} to the codes constructed in Corollary \ref{cor:extended-4}, we obtain the following result.

\begin{thm}\label{thm:Q1} Assume that $n$ is odd and $(n-1)|(q^2-1)$.
Put $k= \lfloor \frac{n+q-1}{q+1} \rfloor$ and $i_{\min}=\min (\lfloor \frac{n-q}{2}\rfloor-k,q-1)$. Assume further that $(n-1)|k(q+1)$. Then, for any $1\le k'\le k+{i_{\min}}+1$ and $0\le \ell\le k'-\sharp \Delta_{k'}$, there exist an $[[n+2, k'-\ell, n+2-k';n+2-k'-\ell]]_{q}$ EAQECC and an $[[n+2, n+2-k'-\ell,  k';k'-\ell]]_{q}$.
\end{thm}

\begin{rem} EAQECCs obtained from Theorem \ref{thm:Q1} have new parameters that can not be constructible by \cite{FangFuLiZhu}, for instance, the code length in Theorem \ref{thm:Q1} can take the value $q^2+2$ while the code lengths \cite{FangFuLiZhu} are always less than or equal to $q^2+1$.
\end{rem}

It is well known that the Hermitian dual of an MDS linear code is again an MDS linear code. By applying Lemma \ref{lem:GJG18-construction} to the MDS $[n,k]_{q^2}$ codes constructed in Theorem \ref{thm:extended-2} and Corollary \ref{cor:extended-3}, we obtain the following result.

\begin{thm}\label{thm:Q2} Assume that $n$ is odd and $(n-1)|(q^2-1)$.
Put $k= \lfloor \frac{n+q-1}{q+1} \rfloor$ and $i_{\min}=\min (\lfloor \frac{n-q}{2}\rfloor-k,q-1)$. Assume further that $(n-1)|k(q+1)$. Then, for any $1\le k'\le k+{i_{\min}}+1$ and $0\le \ell\le k'-\sharp \Delta_{k'}$, 
\begin{enumerate}
\item there exist an $[[n, k'-\ell, n-k'+1;n-k'-\ell]]_{q}$ EAQECC and an MDS $[[n, n-k'-\ell, k'+1;k'-\ell]]_{q}$ EAQECC;
\item there exist an $[[n+1, k'-\ell,  n+2-k';n+1-k'-\ell]]_{q}$ EAQECC and an MDS $[[n+1, n+1-k'-\ell,  k'+1;k'-\ell]]_{q}$ EAQECC.
\end{enumerate}
\end{thm}
\begin{rem} EAQECCs obtained from Theorem \ref{thm:Q2} have new parameters that can not be covered by those in \cite{FangFuLiZhu}, for instance, for $q=7$, 
there exists, in \cite{FangFuLiZhu}, a $[[17,k'-\ell,17-k'+1;17-k'-\ell]]_7$ code for any $1\le k'\le 2$ and for any $0\le \ell \le k'$. However, from Theorem \ref{thm:Q2},
there exists a $[[17,k'-\ell,17-k'+1;17-k'-\ell]]_5$ for any $2\le k'\le 6$ and for any $0\le \ell \le k'-2$.

For the code length $n=q^2$, the code dimensions in \cite{FangFuLiZhu} 
range between $1$ and $q-1$, but our code dimensions can take the values up to $2(q-1)$. For instance, for $q=5$, 
there exists, in \cite{FangFuLiZhu}, a $[[25,k'-\ell,25-k'+1;n-k'-\ell]]_5$ code for any $2\le k'\le 4$ and for any $0\le \ell \le k'$. However, from Theorem \ref{thm:Q2},
there exists a $[[25,k'-\ell,25-k'+1;n-k'-\ell]]_5$ for any $2\le k'\le 9$ and for any $0\le \ell \le k'-2$.
\end{rem}

By applying Lemma \ref{lem:GJG18-construction} to the MDS $[n,k]_{q^2}$ codes constructed in Theorem \ref{thm:extended-3}, we obtain the following result.

\begin{thm}\label{thm:Q3} Let $q=p^m$ be a prime power and $N|(q^2-1)$. Put $n=(t+1)N+1$, $k= \lfloor \frac{n+q-1}{q+1} \rfloor$ and $i_{\min}=\min (\lfloor \frac{n-q}{2}\rfloor-k,q-1)$, where $1\le t\le \frac{q-1}{n_2}-2$ and $n_2=\frac{N}{\gcd (N,q+1)}$. Assume that $n$ is odd and $(n-1)|k(q+1)$ or $(n-1)|(k+1)(q+1)$. Then, for any $1\le k'\le k+{i_{\min}}+1$ and $0\le \ell\le k'-\sharp \Delta_{k'}$, 
\begin{enumerate}
\item there exist an $[[n, k'-\ell, n-k'+1;n-k'-\ell]]_{q}$ EAQECC and an MDS $[[n, n-k'-\ell, k'+1;k'-\ell]]_{q}$ EAQECC;
\item there exist an $[[n+1, k'-\ell,  n+2-k';n+1-k'-\ell]]_{q}$ EAQECC and an MDS $[[n+1, n+1-k'-\ell,  k'+1;k'-\ell]]_{q}$ EAQECC.
\end{enumerate}
\end{thm}
\begin{rem} EAQECCs obtained from Theorem \ref{thm:Q3} have new parameters that can not be covered by those in \cite{FangFuLiZhu}, for instance, for $q=5$, 
there exists, in \cite{FangFuLiZhu}, a $[[19,k'-\ell,19-k'+1;19-k'-\ell]]_5$ code for any $1\le k'\le 3$ and for any $0\le \ell \le k'$. However, from Theorem \ref{thm:Q3}, there exists a $[[19,k'-\ell,19-k'+1;19-k'-\ell]]_5$ for any $3\le k'\le 8$ and for any $0\le \ell \le k'-3$.
\end{rem}



\begin{thebibliography}{99}


\bibitem{AK01}
A. Ashikhmin and E. Knill, ``Nonbinary quantum stabilizer codes," \emph{IEEE Trans. Inf. Theory}, vol. 47, no. 7, pp. 3065-3072, Nov. 2001.

\bibitem{AssKey}  E. F. Assmus and Jr. J.D. Key, ``Affine and projective planes," {\em Discrete Math.} 83, pp. 161--187, 1990.

  \bibitem{BeePuRos17} P. Beelen, S. Puchinger, and J. Rosenkilde, ``Twisted Reed-Solomon Codes," in {\em IEEE ISIT}, 2017, pp. 336-340.
  
  \bibitem{BeePuRos18} P. Beelen, M. Bossert, S. Puchinger, and J. Rosenkilde, ``Structural properties of twisted Reed-Solomon codes
with applications to code-based cryptography," In {\em IEEE ISIT}, pp. 946--950 (2018).
\bibitem{Bie Ede} J.~Bierbrauer and Y.~Edel, ``Quantum twisted
codes," \emph{J. Comb. Designs,}  vol. 8, pp. 174--188, 2000.

\bibitem{Mag} W. Bosma and J. Cannon, {\em   Handbook of Magma Functions}, Sydney, 1995.

\bibitem{Bow} G. Bowen, ``Entanglement required in achieving entanglement-assisted channel capacities," {\em Physical Review A,} 66, 052313--1--052313--8 (Nov 2002).
\bibitem{BDH06} T. Brun, I. Devetak and M.H. Hsieh, ``Correcting quantum errors with entanglement,'' \emph{Science}, vol. 314, pp. 436-439, Oct.
2006.

\bibitem{Cal Rai}A.R. Calderbank, E.M. Rains, P.W. Shor and
N.J.A. Sloane, ``Quantum error correction via codes over GF(4),"
\emph{IEEE Trans. Inf. Theory,} vol. 44, no. 4, pp. 1369--1387,
July 1998.

\bibitem{CarGui} Carlet C., Guilley S.: Complementary dual codes for counter-measures to side-channel attacks. In: E.R. Pinto {et al.} (eds.), {Coding Theory and Applications,} {\em CIM Series in Mathematical Sciences,} 3, pp. 97--105, Springer (2014). {\em Adv. Math. Commun.} 10(1), pp. 131--150, 2016.
\bibitem{CarLiMes} C. Carlet, C. Li and S. Mesnager, ``Linear codes with small hulls in semi-primitive case," {\em Des. Codes Cryptogr.} https://doi.org/10.1007/s10623-019-00663-4
\bibitem{CarMesTanQiPel18} C. Carlet, S. Mesnager, C. Tang, Y. Qi and R. Pellikaan, `` Linear codes over $\F_q$ are equivalent to LCD codes for $q>3,$" {\em IEEE Trans. Inf. Theory,} 64(4), pp. 3010--3017, 2018.

\bibitem{CarMesTanQi18-2} C. Carlet, S. Mesnager, C. Tang and Y. Qi, `` Euclidean and Hermitian LCD MDS codes," {\em Des. Codes Cryptogr.} 86, pp. 2605--2618, 2018.


\bibitem{CL18}
 B. Chen and H. Liu, ``New constructions of MDS codes
with complementary duals,'' \emph{IEEE Trans. Inf. Theory}, vol.
64, no. 8, pp. 5776-5782, Aug. 2018.

\bibitem{FangFu} W. Fang and F.-W. Fu, ``Two new classes of quantum MDS codes," {\em Finite Fields
Appl.}, vol. 53, pp. 85--98, Sep. 2018.
\bibitem{FangFuLiZhu}W. Fang, F.-W. Fu, L. Li and S. Zhu, ``Euclidean and Hermitian Hulls of MDS Codes and Their Applications to EAQECCs," \emph{IEEE Trans. Inf. Theory,} vol. 66(6), pp. 3527--3537, June 2020.

\bibitem{GalHerMatRua} C. Galindo, F. Hernando, R. Matsumoto, D. Ruano, ``Entanglement-assisted quantum error-correcting codes over arbitrary finite fields," {\em Quantum Information Processing,} 18(4), 116 (Apr 2019).

\bibitem{Grassl} M. Grassl, ``Entanglement-Assisted Quantum Communication Beating the Quantum Singleton Bound," Phys. Rev. A
103:020601, 2021.
\bibitem{GHW} M. Grassl, F. Huber, and A. Winter, ``Entropic proofs of Singleton bounds for quantum error-correcting codes," https://arxiv.org/abs/2010.07902

\bibitem{GJG18} K. Guenda, S. Jitman and T.A. Gulliver, ``Constructions of good entanglement assisted quantum error correcting
codes,'' \emph{Des. Codes Cryptogr.}, vol. 86, pp. 121-136, Jan. 2018.


\bibitem{GGJT18}
K. Guenda, T.A. Gulliver, S. Jitman and S. Thipworawimon, ``Linear $\ell$-intersection pairs of codes and their applications,'' \emph{Des. Codes
Cryptogr.}, 2019. https://doi.org/10.1007/s10623-019-00676-z

\bibitem{GraBet} M. Grassl and T. Beth, ``Quantum BCH codes," {\em Proceedings of International Symposium on Theoretical Electrical Engineering Magdeburg}, pp. 207--212, Oct. 1999. DOI:10.1109/ICCES.2008.4772987

\bibitem{Gra Bet2} M. Grassl, T. Beth and M. R\"{o}ttler, ``On optimal quantum
codes," \emph{Int. J. Quantum Inf.,} vol. 2, no. 1, pp. 757--775,
2004.

\bibitem{Gu11}  G.G.L. Guardia, ``New quantum MDS codes," \emph{IEEE Trans. Inform.
Theory,} vol. 57, no. 8, pp. 5551-554,  2011.

\bibitem{HuangYueNiuLi} H. Huang, Q. Yue, Y. Niu, and  X. Li,``MDS or NMDS self-dual codes from twisted
generalized Reed-Solomon codes," {\em Des. Codes Cryptogr.,} 89, 2195-2209
(2021)

\bibitem{JinXin} L. F. Jin and C. P. Xing, ``New MDS self-dual codes from generalized Reed-Solomon codes," {\em IEEE Trans. Inform. Theory}, vol. 63(3) , pp. 1434 --1438, 2017

\bibitem{Ket Kla} A. Ketkar, A. Klappenecker, S. Kumar and P. Sarvepalli, ``Nonbinary stablizer codes over finite fields," \emph{IEEE. Trans. Inform. Theory,} vol. 52, no. 11,  pp.4892--4914, Nov.
2006.

\bibitem{LB13}
C.-Y. Lai and T.A. Brun, ``Entanglement increases the error-correcting
ability of quantum error-correcting codes,'' \emph{Phys. Rev. A, Gen. Phys.},
vol. 88, p. 012320, Jul. 2013.

\bibitem{LA18}
 C.Y. Lai and A. Ashikhmin, ``Linear programming bounds for entanglement-assisted quantum error correcting codes by split weight enumerators,'' \emph{IEEE Trans. Inf. Theory}, vol.
64, no. 1, pp. 622-639, Jan. 2018.


\bibitem{Leon82}J. Leon, ``Computing automorphism groups of error-correcting codes," {\em IEEE Trans. Inf. Theory,} 28(3), pp. 496--511, 1982.
\bibitem{Leon91} J. Leon, ``Permutation group algorithms based on partition, I: Theory and algorithms," {\em J. Symb. Comput.} 12: 533--583, 1991.

\bibitem{Li Xin Wan2} Z. Li, J. Xing and X.M. Wang, ``Quantum generalized Reed-Solomon codes: unified
framework for quantum MDS codes,"  \emph{Phys. Rev. A,} vol. 77, pp. 012308-1--12308-4, 2008.

\bibitem{LiZeng} C. Li and P. Zeng, ``Constructions of linear codes with one-dimensional hull," {\em IEEE Trans. Inf. Theory,} 65 (3), pp. 1668--1676, 2019.

\bibitem{LuoCaoChen} Luo G., Cao X., Chen X.: MDS Codes With Hulls of Arbitrary Dimensions and Their Quantum Error Correction. {\em IEEE Trans. Inf. Theory} 65(5), pp. 2944--2952, 2019.


\bibitem{LCC}
G. Luo, X. Cao and X. Chen, ``MDS codes with hulls of arbitrary dimensions and
their quantum error correction,'' \emph{IEEE Trans. Inf. Theory}, vol. 65, no. 5, pp. 2944-2952, May 2019.


\bibitem{PerPel} F.R.F. Pereira, R. Pellikaan, G.G.L. Guardia, F.M.D. Assis, 
``Entanglement-assisted Quantum Codes from Algebraic Geometry Codes," https://arxiv.org/pdf/1907.06357.pdf

\bibitem{QCM} L. Qian, X. Cao and S. Mesnager, ``Linear codes with one-dimensional hull associated with Gaussian sums", {\em Cryptogr. Commun.} (2020). https://doi.org/10.1007/s12095-020-00462-y

\bibitem{Sendrier00}N. Sendrier, ``Finding the permutation between equivalent codes: the support splitting algorithm," {\em IEEE Trans. Inf. Theory,} 46(4), pp. 1193--1203, 2000.
\bibitem{Sok1D} L. Sok, ``MDS linear codes with one dimensional hull," {\it Cryptogr. Commun.}, DOI: 10.1007/s12095-022-00559-6
\bibitem {Sok1D2} L. Sok, ``On linear codes with one-dimensional Euclidean hull and their applications to EAQECCs," {\em IEEE Trans. Inf. Theory,} DOI: 10.1109/TIT.2022.3152580
\bibitem{SokQSC} L. Sok, ``New families of quantum stabilizer codes from Hermitian self-orthogonal algebraic geometry codes", submitted for publication, available at https://arxiv.org/


\bibitem{WilBru} M.M. Wilde and T.A. Brun, ``Optimal entanglement formulas for entanglement-assisted quantum coding," {\em Physical Review A,} 77(6), 064302--1--064302--4 (Jun 2008).
\end{thebibliography}
\end{document}